\documentclass[conference]{IEEEtran}

\usepackage{epsfig,endnotes}
\usepackage{multirow}
\usepackage{subfig}
\usepackage{graphicx}
\usepackage{grffile}

\newcounter{subcopyrightbox@save}
\usepackage[font=bf]{caption}
\usepackage{color, url}
\usepackage{xspace} 
\usepackage{lipsum}
\usepackage{multicol}

\usepackage{mathrsfs}
\usepackage{amssymb}
\usepackage{amsmath}
\usepackage{amsthm}
\usepackage{epstopdf}
\usepackage{balance}
\usepackage{bbm}
\newtheorem{theorem}{Theorem}
\newtheorem{lemma}[theorem]{Lemma}
\newtheorem{definition}[theorem]{Definition}
\usepackage{algorithm}
\usepackage{algorithmic}

\newcommand{\argmax}{\operatornamewithlimits{argmax}}
\newcommand{\argmin}{\operatornamewithlimits{argmin}}

\newcommand{\myparatight}[1]{\smallskip\noindent{\bf {#1}:}~}

\graphicspath{ {../fig/} }

\begin{document}

\title{Calibrate: Frequency Estimation and Heavy Hitter Identification with Local Differential Privacy via Incorporating Prior Knowledge}

\author{
    \IEEEauthorblockN{Jinyuan Jia, Neil Zhenqiang Gong}
    \IEEEauthorblockA{Department of Electrical and Computer Engineering, Iowa State University}
    \IEEEauthorblockA{\{jinyuan, neilgong\}@iastate.edu}
}


\maketitle

\begin{abstract}
Estimating frequencies of certain items among a population is a basic step in data analytics, which enables more advanced data analytics (e.g., heavy hitter identification, frequent pattern mining), client software optimization, and detecting unwanted or malicious hijacking of user settings in browsers. Frequency estimation and heavy hitter identification with \emph{local differential privacy (LDP)} protect user privacy as well as the data collector. Existing LDP algorithms cannot leverage 1) prior knowledge about the noise in the estimated item frequencies and 2) prior knowledge about the true item frequencies. As a result, they achieve suboptimal performance in practice. 

In this work, we aim to design LDP algorithms that can leverage such prior knowledge. Specifically, we design ${Calibrate}$ to incorporate the prior knowledge via statistical inference. ${Calibrate}$ can be appended to an existing LDP algorithm to reduce its estimation errors. We model the prior knowledge about the noise and the true item frequencies as two probability distributions, respectively. Given the two probability distributions and an estimated frequency of an item produced by an existing LDP algorithm, our ${Calibrate}$ computes the conditional probability distribution of the item's  frequency and uses the mean of the conditional probability distribution as the calibrated frequency for the item. It is challenging to estimate the two probability distributions due to data sparsity. We address the challenge via integrating techniques from statistics and machine learning. Our empirical results on two real-world datasets show that ${Calibrate}$ significantly outperforms state-of-the-art LDP algorithms for frequency estimation and heavy hitter identification. 

\end{abstract}

\section{Introduction}
In \emph{frequency estimation}, a \emph{data collector} aims to estimate the frequencies of certain \emph{items} among a population, where frequency of an item is the number of users that have the item. 
Frequency estimation is a basic research problem in data analytics and networking services. 
For instance, 
Google may be interested in estimating how many users set a particular webpage as the default 
homepage of Chrome~\cite{Erlingsson:2014}, where Google is the data collector and each webpage is an item; and an app developer may be interested in estimating how many users adopt a certain feature of the app, where the app developer is the data collector and each feature of the app is an item. Such {frequency estimation} is often the first step to perform more advanced data analytics, optimize client software (e.g., web services, mobile apps), and detect unwanted or malicious hijacking of user settings in browsers~\cite{Erlingsson:2014}. For instance, after estimating item frequencies, the data collector can identify the items whose frequencies are larger than a given threshold, which is called \emph{heavy hitter identification}. 

A naive solution for frequency estimation or heavy hitter identification is to ask each user to share its item or set of items with the {data collector}, who can compute the items' frequencies easily. However, this naive solution faces two challenges. First, when the items are sensitive, users may not be willing to share their raw items with the data collector. Second, the data collector could be vulnerable to insider attacks and could be compromised to leak the users' items, which frequently happens in real world, e.g., Equifax was recently compromised and personal data of 143 million users were leaked~\cite{equifaxdataleak}.  

Local differential privacy (LDP)~\cite{Erlingsson:2014,BassilySuccinctHistograms15,kairouz2016discrete,QinHeavyHitterCCS16, Ninghui:2017, AppLDP17,ding2017collecting,duchi2013local,ye2017optimal,wang2017local,qin2017generating,bun2018heavy,smith2017interaction,duchi2013localb,bassily17pratical,wang2018local,wang2018itemset,zhang2018calm,acharya2018hadamard}, a privacy protection mechanism based on $\epsilon$-differential privacy~\cite{Dwork:2006}, can address both challenges. Several algorithms~\cite{Erlingsson:2014,BassilySuccinctHistograms15,kairouz2016discrete,QinHeavyHitterCCS16,Ninghui:2017,AppLDP17,ding2017collecting} for frequency estimation with LDP have been proposed recently. Moreover, LDP was deployed by Google Chrome~\cite{Erlingsson:2014}, Apple~\cite{AppLDP17}, Microsoft~\cite{ding2017collecting}, and Uber~\cite{UberDP}. In fact, LDP is the first privacy mechanism that was widely deployed in industry. These LDP algorithms essentially consist of three steps, i.e., $Encode$, $Perturb$, and $Aggregate$~\cite{Ninghui:2017}. $Encode$ and $Perturb$ are executed at client side for each user, while $Aggregate$ is executed at the data collector side.  
$Encode$ translates a user's item into a number or vector; $Perturb$ perturbs a user's encoded number or vector to preserve LDP, and sends the perturbed value to the data collector;  $Aggregate$ estimates item frequencies using the perturbed values from all users.  

However, existing LDP algorithms have two key limitations. Specifically, the estimated frequency of an item produced by the $Aggregate$ step is the sum of the true item frequency and some noise. First, existing LDP algorithms do not leverage the prior knowledge about the noise to filter them. Second, in many scenarios, the data collector could have prior knowledge about the true item frequencies, but existing LDP algorithms cannot leverage such prior knowledge. For instance, many real-world phenomena such as video popularity, webpage click frequency, node degrees in social networks--follow power-law distributions~\cite{cha2007tube,Clauset09,Gong12-imc}; human height follows a Gaussian distribution~\cite{a2009height}. In a hybrid LDP setting~\cite{Blender:2017}, some \emph{opt-in} users share their true items with the data collector, who could obtain prior knowledge about the distribution of the true item frequencies from such {opt-in} users. Moreover, hypothesis testing with LDP~\cite{gaboardi2017local} can also help to determine the distribution family of the true item frequencies.  Due to these two limitations, existing LDP algorithms achieve suboptimal performance. 

\begin{figure}[!t]
\centering
\subfloat{\includegraphics[width= 0.45\textwidth]{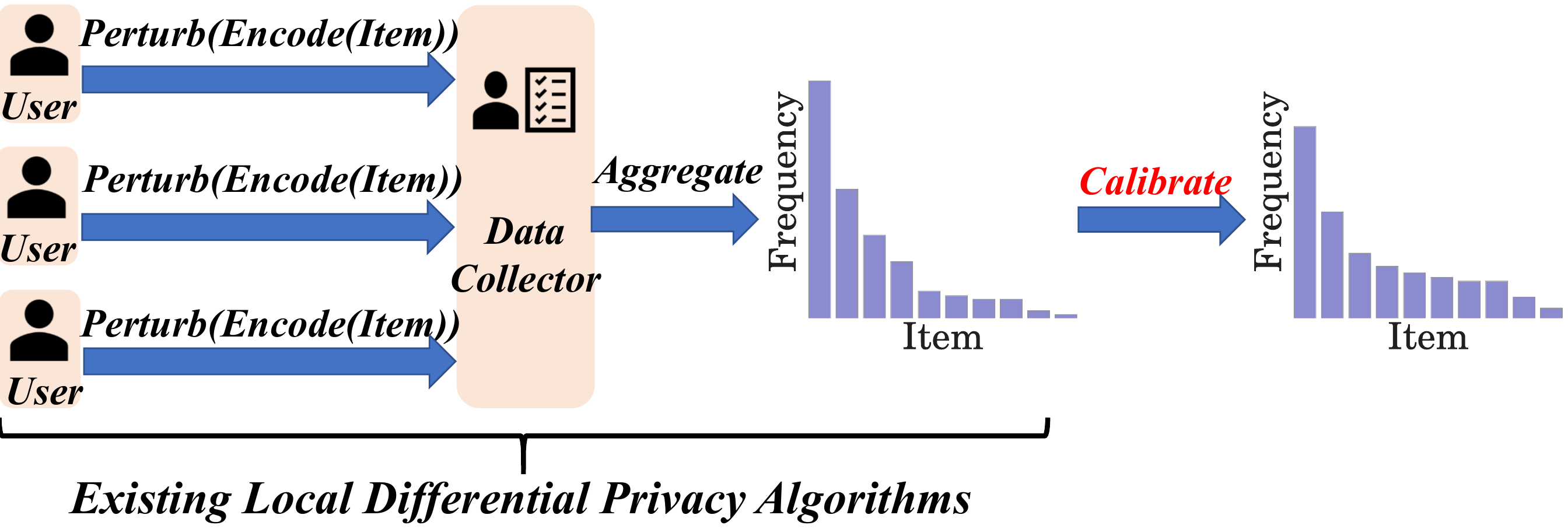}}
\caption{Our statistical inference framework.}
\vspace{-2mm}
\label{algorithma}
\end{figure}

\myparatight{Our work} In this work, we aim to address the two limitations of existing LDP algorithms. Towards this goal, we propose a $Calibrate$ step, which incorporates the prior knowledge about the noise and the true item frequencies via statistical inference. 
Our $Calibrate$ step can be appended to the $Aggregate$ step to reduce the noise and increase accuracies for an existing LDP algorithm. Figure~\ref{algorithma} illustrates existing LDP algorithms and our $Calibrate$ step. 
 We model the prior knowledge about the noise and the true item frequencies as probability distributions $p_s$ and $p_f$, respectively. The two probability distributions mean that, if we sample an item uniformly at random, then 1) the item's estimated frequency produced by an existing LDP algorithm has a noise $\delta$ with a  probability $p_s(\delta)$, and 2) the item's true frequency is $k$ with a probability $p_f(k)$. 
Given the two probability distributions and the estimated frequency of an item produced by a certain LDP algorithm, we compute the conditional probability distribution of the item's frequency, and we use the mean of the conditional probability distribution as a calibrated frequency for the item. We theoretically show that, under the same conditions (i.e., given the two probability distributions and an estimated item frequency), our $Calibrate$ is the \emph{optimal} estimator to refine the item frequency. 

Implementing our $Calibrate$ step faces two challenges. First, how to estimate the probability distribution $p_s$ of the noise? Second, how to estimate the probability distribution $p_f$ of the true item frequencies. To address the first challenge, we theoretically show that, for state-of-the-art LDP algorithms, the noise follows a Gaussian distribution with mean 0 and a known variance. To address the second challenge, we consider the data collector knows the distribution family that $p_f$ belongs to.  A distribution family is a group of probability distributions that have similar properties, e.g., power-law distribution is a distribution family that we would observe in many real-world phenomena. A distribution family is parameterized by certain parameters, e.g., a power-law distribution is parameterized by an exponent. 
We design a \emph{mean-variance method} to estimate the parameters in the probability distribution $p_f$ efficiently.

We perform extensive experiments on two real-world datasets to evaluate our $Calibrate$ step and compare it with state-of-the-art LDP algorithms for frequency estimation and heavy hitter identification. Our results demonstrate that once we append our $Calibrate$ step to an existing LDP algorithm, we can increase accuracies significantly.  

In summary, our contributions are as follows:
\begin{itemize}
\item We propose a statistical inference framework called $Calibrate$ to incorporate prior knowledge about 1) noise in the estimated item frequencies and 2) true item frequencies. Our $Calibrate$ can be appended to an existing LDP algorithm to improve accuracies. 

\item We design methods to estimate the probability distributions that model such prior knowledge. 

\item We perform extensive experiments on two real-world datasets to evaluate $Calibrate$. Our results show that $Calibrate$ significantly outperforms state-of-the-art LDP algorithms for frequency estimation and heavy hitter identification. 

\end{itemize}

\section{Background and Related Work}
\label{relatedwork}

\subsection{Frequency Estimation and Heavy Hitter Identification}
\myparatight{Frequency estimation} Suppose we have $d$ items (denoted as $\{1,2,\cdots,d\}$) and $n$ users. Each user has an item.\footnote{Our techniques can also be applied to the scenarios where each user has a set of items.}   
A data collector aims to compute the frequency $f_i$ of each item among the $n$ users, i.e., $f_i$ is the number of users that have the item $i$.  
For instance,  an item is Yes or No when the data collector is interested in estimating the number of users that are HIV positive and negative in a survey; 
an item is a webpage when the data collector (e.g., Google) aims to estimate the number of users that set a particular webpage as a browser's default homepage; an item is a feature of a mobile app when the data collector (e.g., the app developer) aims to estimate the number of users that use each feature. 

\myparatight{Heavy hitter identification} A direct application of frequency estimation is heavy hitter identification. Specifically, given a threshold, heavy hitter identification aims to detect the items whose frequencies are larger than the threshold. Heavy hitter identification
is a basic research problem in data analytics with many applications, such as trend monitoring, marketing analysis, and anomaly
detection. 



%


\subsection{Local Differential Privacy Algorithms}
The first  local differential privacy (LDP) algorithm for frequency estimation called \emph{randomized response} dates back to 1960s~\cite{RR1965}. 
Recently, several LDP algorithms for frequency estimation~\cite{Erlingsson:2014,kairouz2016discrete,QinHeavyHitterCCS16,Ninghui:2017,AppLDP17,ding2017collecting} were proposed, e.g., RAPPOR~\cite{Erlingsson:2014},  $k$-RR~\cite{kairouz2016discrete},  and Optimized Unary Encoding (OUE)~\cite{Ninghui:2017}.  An existing LDP algorithm $\mathcal{A}$ essentially consists of three functions, i.e., $\mathcal{A}=(Encode, Perturb, Aggregate)$. Figure~\ref{algorithma} illustrates the three key functions for frequency estimation with LDP. The $Encode$ function encodes a user's  item into a numerical value or a vector;  the $Perturb$ function perturbs a user's encoded value or vector such that local differential privacy is achieved; the $Aggregate$ function estimates item frequencies from the perturbed encoded values or vectors from all users. The $Encode$ and $Perturb$ functions are performed at client side for every user, while the $Aggregate$ function is executed at the data collector side. For simplicity, we denote $Perturb(Encode(i))$ as $PE(i)$, where $i$ is an item. 




  
Roughly speaking, in LDP, any two items have close probabilities to be mapped to the same perturbed numerical value or vector. 
Moreover, state-of-the-art LDP algorithms (e.g., basic RAPPOR~\cite{Erlingsson:2014}, $k$-RR~\cite{kairouz2016discrete}, and OUE~\cite{Ninghui:2017}) achieve \emph{pure local differential privacy}~\cite{Ninghui:2017}, which is formally defined as follows:
\begin{definition}[{\bf Pure Local Differential Privacy}]
A randomized algorithm $\mathcal{A}=(Encode, Perturb, Aggregate)$ achieves pure local differential privacy if and only if there exists two probability values $p^*$ and $q^*$ such that for all item  $i$.
\begin{align}
&Pr[PE(i)) \in \{t|i\in Support(t)\}]=p^*, \\
\forall_{j\neq i}&Pr[PE(j) \in \{t| i \in Support(t)\}]=q^*,
\end{align}
where $Support(t)$ is the set of items that a perturbed numerical value or vector $t$ supports. 
\end{definition}
For any pure LDP algorithm, the data collector can use the following equation to estimate item frequencies in the Aggregate function~\cite{Ninghui:2017}:
\begin{align}
\label{agg}
\hat{f}_i=\frac{\sum_{u}\mathbbm{1}_{Support(t_u)}(i)-nq^*}{p^*-q^*},
\end{align}
where $\hat{f}_i$ is the estimated frequency for item $i$, $t_u$ is the perturbed encoded output of user $u$, and the function $\mathbbm{1}_{Support(t_u)}(i)$ is 1 if $t_u$ supports the item $i$, otherwise it is 0. Intuitively, for every perturbed encoded output from users, we add one count to the items that are supported by the output. In the end, we normalize the counts using the probabilities $p^*$ and $q^*$. Different pure LDP algorithms use different $Encode$ and/or $Perturb$ functions, and thus they have different $Support$ functions, $p^*$, and $q^*$.  

Next, we use a state-of-the-art pure LDP algorithm called OUE~\cite{Ninghui:2017} as an example to illustrate the three functions. OUE is an optimized version of the basic RAPPOR algorithm proposed by Erlingsson et. al.~\cite{Erlingsson:2014}. 

\myparatight{Encode} OUE uses \emph{unary encoding} to encode an item. Specifically, OUE uses a length-$d$ binary vector $\mathbf{X}$ to encode the $d$ items. If a user has item $i$, then $\mathbf{X}[i]=1$ and all other entries of $\mathbf{X}$ are 0.   

\myparatight{Perturb} In this step, OUE perturbs a binary vector $\mathbf{X}$ into another binary vector $\mathbf{Y}$ probabilistically bit by bit. Specifically, we have:
\begin{align}
\label{oueperturb}
\text{Pr}(\mathbf{Y}[i]=1)=\left \{ 
\begin{aligned}
p=&\frac{1}{2}, \text{ if }\mathbf{X}[i]=1 \\
q=&\frac{1}{1+e^{\epsilon}}, \text{ if }\mathbf{X}[i]=0 
\end{aligned}
\right.
\end{align}
If a bit in $\mathbf{X}$  is $1$, then the corresponding bit in $\mathbf{Y}$ will be 1 with a probability $\frac{1}{2}$. However, if a bit in $\mathbf{X}$  is $0$, then the corresponding bit in $\mathbf{Y}$ will be 1 with a  probability $\frac{1}{1+e^{\epsilon}}$, where $\epsilon$ is the privacy budget. 

\myparatight{Aggregate} Once the data collector receives the perturbed binary vectors $\mathbf{Y}$ from users, the data collector estimates the frequency $\hat{f}_i$ of item $i$ using Equation~\ref{agg} with the function $\mathbbm{1}_{Support(t_u)}(i)=\mathbf{Y}_{u}[i]$, $p^*=p$, and $q^*=q$, where $\mathbf{Y}_{u}$ is the perturbed binary vector from user $u$.

After estimating the item frequencies, we can identify heavy hitters. We note that some studies~\cite{Hsu12,BassilySuccinctHistograms15,QinHeavyHitterCCS16,bassily17pratical,wang2017locally} designed LDP algorithms to identify top-$k$ heavy hitters, which are the $k$ items that have the largest frequencies. However, such algorithms cannot identify threshold-based heavy hitters, because they do not consider the item frequencies.

\subsection{Evaluation Metrics} 

\myparatight{Frequency estimation} Since a LDP algorithm is a randomized algorithm, the estimated frequency $\hat{f}_i$ is a random variable, which means that every time the data collector executes the LDP algorithm, the estimated frequency $\hat{f}_i$ could be different. 
Therefore, like previous studies~\cite{Ninghui:2017,Erlingsson:2014,BassilySuccinctHistograms15,kairouz2016discrete}, we use the \emph{mean square error (MSE)} of the random variable $\hat{f}_i$ to measure error of a LDP algorithm at estimating the frequency of item $i$.  Specifically, the MSE for item $i$ is defined as $MSE(\hat{f}_i, f_i)=E(\hat{f}_i - f_i)^2$, where $f_i$ is the true frequency of item $i$, $E$ represents expectation, and the expectation is taken with respect to the probability distribution of $\hat{f}_i$.  Note that if $\hat{f}_i$ is an unbiased estimator, which means that its expectation equals $f_i$, then the MSE for item $i$ is the variance of the random variable $\hat{f}_i$.
Moreover, the estimation error of a LDP algorithm $\mathcal{A}$ is defined as the average MSE of estimating frequencies of the $d$ items. Formally, estimation error of an algorithm $\mathcal{A}$ is computed as follows: 
\begin{align}
\label{error}
\text{\bf Estimation Error of } \mathcal{A}=\frac{1}{d}\sum_{i=1}^d MSE(\hat{f}_i, f_i).
\end{align}

\myparatight{Heavy hitter identification} Given a threshold, we define an item as a True (False) Positive if the item has a true frequency larger (smaller) than the threshold and is estimated to have a frequency larger than the threshold. We also define an item as a False Negative if the item has a true frequency larger than the threshold but is estimated to have a frequency smaller than the threshold. We use standard metrics in information retrieval to measure the quality of heavy hitter identification: 
\begin{align}
\text{Precision} &= \frac{\text{True Positive}}{\text{True Positive + False Positive}} \\
\text{Recall} &= \frac{\text{True Positive}}{\text{True Positive + False Negative}} \\
\text{F-Score} &= \frac{2\cdot \text{Precision} \cdot \text{Recall}}{\text{Precision + Recall}}  
\end{align}

\section{Our $Calibrate$ Framework}

Existing LDP algorithms consist of three steps $Encode$, $Perturb$, and $Aggregate$. 
Our $Calibrate$ can be appended to an existing LDP algorithm as the fourth step. 

\subsection{Overview of $Calibrate$}
\myparatight{Formulating $Calibrate$ as an optimization problem} Suppose we have an LDP algorithm $\mathcal{A}$. The data collector estimates the frequency of an item $i$ to be $\hat{f}_i$ via executing the algorithm.  $\hat{f}_i$ is a sum of the true frequency $f_i$ of the item $i$ and a noise. Specifically, we split the estimated frequency $\hat{f}_i$ as follows:
\begin{align}
\label{relation}
\hat{f}_i=f_i+s_i, \  i \in \{1, 2, \cdots, d\},
\end{align}
where $s_i$ is a noise. We model the noise as a random variable $s$, whose probability distribution is denoted as $p_s$. We view the $d$ noise $s_1, s_2, \cdots, s_d$ as random samples from the probability distribution $p_s$. We model the true item frequency as a random variable $f$, whose probability distribution is denoted as $p_f$. We view the $d$ true item frequencies $f_1, f_2, \cdots, f_d$ as random samples from the probability distribution $p_f$. We model the estimated item frequency as a random variable $\hat{f}$, whose probability distribution is denoted as $p_{\hat{f}}$. We view the $d$ estimated item frequencies $\hat{f}_1, \hat{f}_2, \cdots, \hat{f}_d$ as random samples from the probability distribution $p_{\hat{f}}$. 
Specifically, if we sample an item $i$ uniformly at random from the $d$ items, then the item has a true frequency of $k$ with a probability of $p_f(f=k)$, the item has an estimated frequency of $\hat{k}$ with a probability of $p_{\hat{f}}(\hat{f}=\hat{k})$, and the noise has a value of $\delta$ with a probability of $p_s(s=\delta)$.

The probability distributions $p_s$ and $p_f$ model the prior knowledge about the noise and the true item frequencies, respectively. 
For instance,  as we will demonstrate in Section~\ref{estimateps}, for pure LDP algorithms, $p_s$ can be well approximated as a Gaussian distribution with a known mean and variance; for many application domains (e.g., video popularity, webpage click frequency, and node degrees in social networks~\cite{cha2007tube,Clauset09}), $p_f$ can be parameterized as a power-law distribution, though the parameters in the power-law distribution have to be estimated from the observed estimated item frequencies $\hat{f}_1, \hat{f}_2, \cdots, \hat{f}_d$.


Given Equation~\ref{relation}, we model the relationships between the random variables $s$, $f$, and $\hat{f}$ as $\hat{f} = f + s$. In other words, we model frequency estimation with LDP as a probabilistic generative process: for a randomly sampled item, the item's true item frequency is sampled from the probability distribution $p_f$, a noise is sampled from the probability distribution $p_s$, and an existing LDP algorithm estimates the item's frequency as the sum of the true frequency and the noise. 

In the probabilistic generative process, we observe an item's estimated frequency produced by an existing LDP algorithm. Our $Calibrate$ step aims to ``reverse'' the generative process to find the true item frequency.  Specifically, given a frequency estimation $\hat{f}_i$ and the three probability distributions $p_s$, $p_f$, and $p_{\hat{f}}$, our $Calibrate$ step aims to produce a calibrated frequency estimation $\tilde{f}_i$, such that the MSE is minimized. More formally, we aim to obtain  $\tilde{f}_i$ via solving the following optimization problem: 
\begin{align}
\label{optimization}
&\tilde{f}_i = \argmin_{f'} E((f' - f)^2 | \hat{f} = \hat{f}_i)  \nonumber \\
\text{subject to: } &\hat{f}_i= f + s, \nonumber \\
			&f \varpropto p_f, \nonumber \\
			&s \varpropto p_s, 
\end{align}
where the expectation is taken with respect to the random variable $f$ conditioned on that the estimated item frequency is $\hat{f} = \hat{f}_i$, $f \varpropto p_f$ means that $f$ is a random variable whose probability distribution is $p_f$, and $s \varpropto p_s$ means that $s$ is a random variable whose probability distribution is $p_s$. 

%

\myparatight{Conditional expectation as an optimal solution to the optimization problem}
Given the estimated item frequency $\hat{f}_i$ and the probability distributions $p_s$, $p_f$, and $p_{\hat{f}}$, we can compute a \emph{conditional probability distribution} of the random variable $f$, which models the knowledge we have about the true frequency of item $i$ after observing the estimated item frequency $\hat{f}_i$. Specifically, according to the Bayes' rule~\cite{bayes1763}, we have:
\begin{align}
\text{Pr}(f=k|\hat{f} = \hat{f}_i) &= \frac{\text{Pr}(f=k, \hat{f} = \hat{f}_i)}{\text{Pr}(\hat{f} = \hat{f}_i)} \nonumber \\
&=\frac{\text{Pr}(\hat{f} = \hat{f}_i|f=k)\text{Pr}(f=k)}{\text{Pr}(\hat{f} = \hat{f}_i)} \nonumber \\
\label{post}
&=\frac{p_s(s=\hat{f}_i - k)p_f(f=k)}{p_{\hat{f}}(\hat{f} = \hat{f}_i)},
\end{align}
where $p_s(s=\hat{f}_i - k)$ is the probability that the noise $s$ is $\hat{f}_i - k$.

Our $Calibrate$ step computes the expectation of the conditional probability distribution in Equation~\ref{post} as the calibrated frequency estimation of item $i$. Formally, 
$Calibrate$ estimates $\tilde{f}_i$ as follows:
\begin{align}
\label{bayes}
\tilde{f}_i = \sum_{k} k\cdot \text{Pr}(f=k|\hat{f} = \hat{f}_i).
\end{align}

We show that our conditional expectation based calibrator in Equation~\ref{bayes} is an optimal solution to the optimization problem in Equation~\ref{optimization} as follows: 
\begin{theorem}
\label{optimality}
Our conditional expectation based calibrator in Equation~\ref{bayes} achieves the minimum MSE $E((f' - f)^2|\hat{f} = \hat{f}_i)$ among all calibrators $f'$. Specifically, we have:
\begin{align}
E((\tilde{f}_i - f)^2 | \hat{f} = \hat{f}_i) \leq E((f' - f)^2 | \hat{f} = \hat{f}_i), \text{ for all } f'. \nonumber
\end{align}

\end{theorem}
\begin{proof}
See Appendix~\ref{prooftheorem1}.
\end{proof}


\myparatight{Relationship and difference with Bayesian inference} Using the terminology of the standard Bayesian inference, it seems like that the probability distribution $p_f$ could be interpreted as a prior probability distribution of the true item frequency and the conditional probability distribution of the random variable $f$, which is shown in Equation~\ref{post}, could be interpreted as the posterior probability distribution of the item frequency after observing the estimated item frequency. However, the key difference with the standard Bayesian inference is that the prior probability distribution is independent from the observed data (i.e., the observed estimated item frequencies in our problem) in standard Bayesian inference, while we estimate the parameters in $p_f$ using the observed estimated item frequencies, i.e., $p_f$ is a \emph{data-dependent prior}.

We note that any post-processing of a differential privacy algorithm also achieves differential privacy with the same privacy guarantee~\cite{Dwork:2006}. Therefore, $Calibrate$ does not sacrifice  privacy guarantees. Next, we will discuss how to estimate the probability distributions $p_s$, $p_f$, and $p_{\hat{f}}$.

\subsection{Estimating $p_s$}
\label{estimateps}
Every time the data collector executes a LDP algorithm, we will have $d$ noise $s_1, s_2, \cdots, s_d$ and $d$ frequency estimations $\hat{f}_1, \hat{f}_2, \cdots, \hat{f}_d$. $p_s$ is the probability distribution formed by the $d$ noise in a single execution trial. Since the LDP algorithm is a randomized algorithm, $s_i$ and  $\hat{f}_i$ across different execution trials are different, even if each user has the same item in different execution trials. For simplicity, we model $s_i$ and  $\hat{f}_i$ as random variables, where the randomness comes from the LDP algorithm and $i=1, 2, \cdots, d$. Moreover, we denote by $s_i^{(j)}$ and $\hat{f}_i^{(j)}$ the 
noise and estimated frequency for item $i$ in the $j$th execution trial, respectively. $s_i^{(1)}, s_i^{(2)}, \cdots$ are random samples from the random variable $s_i$, while $\hat{f}_i^{(1)}, \hat{f}_i^{(2)}, \cdots$ are random samples from the random variable $\hat{f}_i$. 

We note that executing a LDP algorithm multiple trials may compromise user privacy because the noise may be canceled out via aggregating results in multiple execution trials. Memoization~\cite{Erlingsson:2014,ding2017collecting} was proposed to preserve privacy when the data collector repeatedly executes the LDP algorithm to collect data. In memoization, the client side pre-computes each user's perturbed and encoded item and responds to the data collector with the pre-computed value in different execution trials. If the memoization is adopted and users' items do not change in different execution trials, then the noise $s_i$ and estimated item frequency $\hat{f}_i$  are the same in different execution trials. However, to illustrate the randomness of $s_i$ and  $\hat{f}_i$, we assume the memoization is not adopted.

\myparatight{Probability distribution of the noise for an item across multiple execution trials}
State-of-the-art LDP algorithms~\cite{Erlingsson:2014,kairouz2016discrete,Ninghui:2017} satisfy pure local differential privacy~\cite{Ninghui:2017}. Therefore, we will focus on pure LDP algorithms. Pure LDP algorithms estimate $\hat{f}_i$ using Equation~\ref{agg}. The variable $\mathbbm{1}_{Support(t_u)}(i)$ in the Equation~\ref{agg} is a binary random variable, due to the randomness of the LDP algorithm. Therefore,  $\hat{f}_i$ is essentially a sum of $n$ binary random variables (with some normalization). According to the Central Limit Theorem~\cite{degroot2011probability}, $\hat{f}_i^{(1)}, \hat{f}_i^{(2)}, \cdots$ obtained in multiple execution trials of the LDP algorithm approximately form a Gaussian distribution. Moreover, the expectation of $\hat{f}_i$ across multiple execution trials of the LDP algorithm is $f_i$, the true frequency of item $i$; and the variance of $\hat{f}_i$ is approximated as $\frac{nq^*(1-q^*)}{(p^*-q^*)^2}$~\cite{Ninghui:2017}. 

Since $s_i^{(j)}=\hat{f}_i^{(j)} - f_i$, we obtain that $s_i^{(1)}, s_i^{(2)}, \cdots$ obtained in multiple execution trials also form a Gaussian distribution. Moreover, we have the expectation and variance of $s_i$ as follows:
\begin{align}
E(s_i) &= E(\hat{f}_i) - f_i=0 \\
Var(s_i) &= Var(\hat{f}_i) = \frac{nq^*(1-q^*)}{(p^*-q^*)^2}
\end{align}
The expectation and variance do not depend on the item index $i$. 
Therefore, the random noise $s_i$ for each item approximately follows the same Gaussian distribution. In Figure~\ref{noisematrix}, we put the noise of each item in multiple execution trials in a matrix, where the $i$th row corresponds to the noise of item $i$ in different execution trials and the $j$th column corresponds to noise of the $d$ items in the $j$th execution trial. For pure LDP algorithms, the numbers in each row of the matrix are sampled from the same Gaussian distribution.

\begin{figure}[!t]
\centering
\subfloat{\includegraphics[width= 0.35\textwidth]{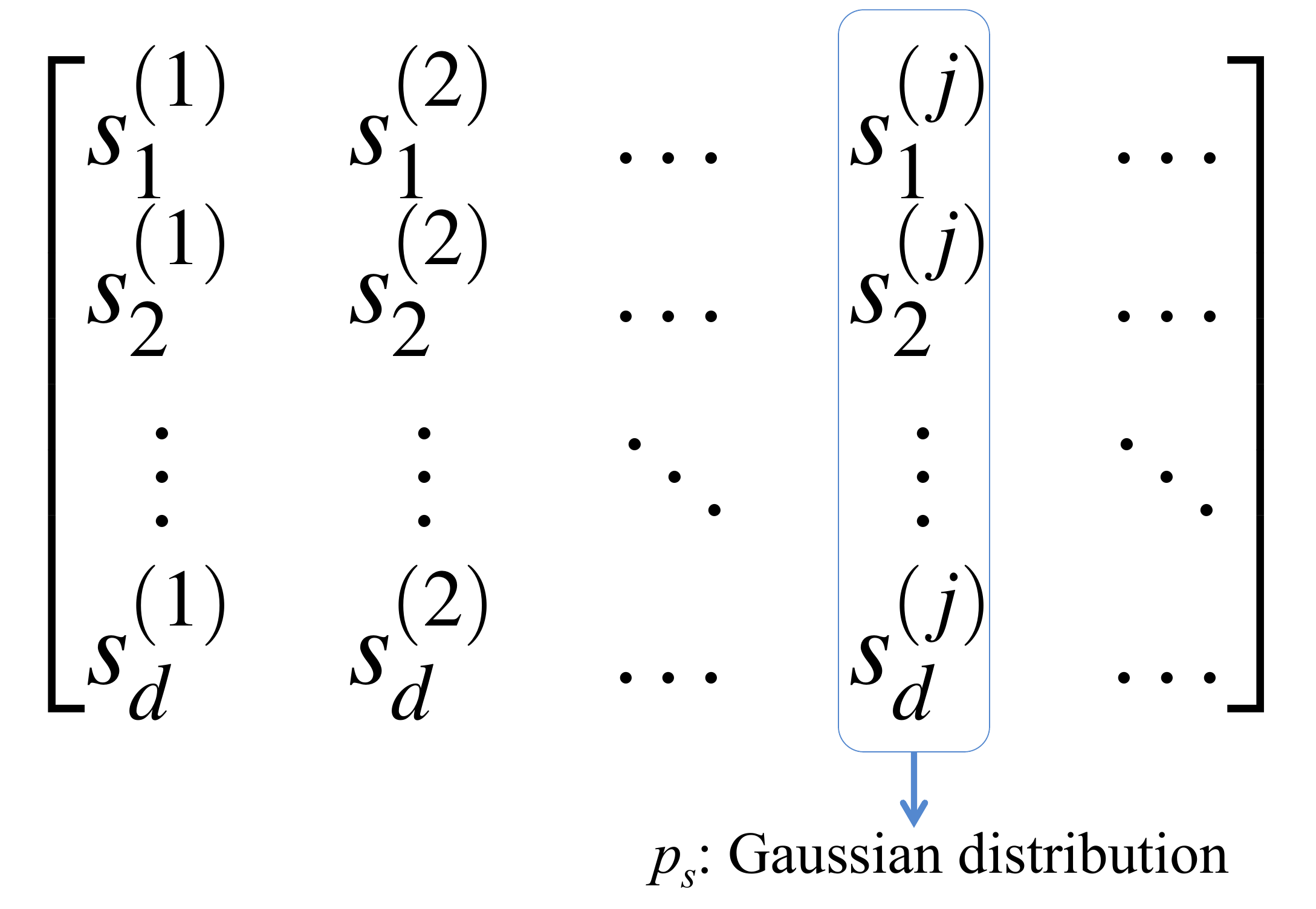}}
\caption{Matrix of noise in multiple execution trials.}
\vspace{-4mm}
\label{noisematrix}
\end{figure}

\myparatight{Probability distribution $p_s$ formed by the noise for all items in a single execution trial} In our formulation of $Calibrate$,  $p_s$ models the probability distribution of the noise of the $d$ items in a single execution trial. Specifically, for the $j$th execution trial, $p_s$ models the probability distribution formed by the noise $s_1^{(j)}, s_2^{(j)}, \cdots, s_d^{(j)}$, i.e., the $j$th column of the matrix illustrated in Figure~\ref{noisematrix}. Since all numbers in the matrix are sampled from the same Gaussian distribution, each column of the matrix follows the same Gaussian distribution. Therefore, $p_s$ is a Gaussian distribution. Moreover, the Gaussian distribution has a mean 0 and variance $\frac{nq^*(1-q^*)}{(p^*-q^*)^2}$. The data collector can compute the variance using $n$, the number of users, and $\epsilon$, the predefined privacy budget ($p^*$ and/or $q^*$ depends on $\epsilon$).
In our experiments, we will show results to empirically verify that $p_s$ is a Gaussian distribution with the known mean and variance.

We note that for LDP algorithms that are not pure LDP, the noise distribution $p_s$ does not necessarily follow a Gaussian distribution. For instance, RAPPOR~\cite{Erlingsson:2014} with Bloom filters is not a pure LDP algorithm (basic RAPPOR without Bloom filters is pure LDP). We empirically found that $p_s$ does not follow a Gaussian distribution. Moreover, $p_s$ depends on the items' true frequencies. For such LDP algorithms with data-dependent noise distribution, our $Calibrate$ is not applicable. However, this limitation is minor since we aim to advance state-of-the-art LDP algorithms, which satisfy pure LDP.

\subsection{Estimating $p_f$ and $p_{\hat{f}}$}
\label{estimatepf}

We assume the data collector knows the distribution family that $p_f$ belongs to. For instance, many real-world phenomena--such as video popularity, webpage click frequency, word frequency in documents, node degrees in social networks--follow power-law distributions~\cite{cha2007tube,Clauset09}; height of human follows a Gaussian distribution~\cite{a2009height}. 
 Moreover, in a hybrid local differential privacy setting~\cite{Blender:2017}, some \emph{opt-in} users trust the data collector and share their true items with the data collector. The data collector could leverage such {opt-in} users to roughly estimate the distribution family that $p_f$ belongs to. Moreover, we design a \emph{mean-variance method} to estimate $p_f$. 
 Based on $p_s$ and $p_f$, we further estimate $p_{\hat{f}}$. 

\myparatight{Estimating $p_f$} Suppose the distribution family of $p_f$ is parameterized by a set of parameters $\Theta$, which we denote as $p_f(f|\Theta)$. For instance, the following shows the popular power-law distribution family: 
\begin{align}
\label{powerlaw}
\text{\bf Power-law: } &p_f(f=k|\alpha) \propto k^{-\alpha} 
\end{align}
We discuss two methods, \emph{maximum likelihood estimation method} and \emph{mean-variance method}, to estimate the parameters $\Theta$. The maximum likelihood estimation method is a standard technique in statistics, while the mean-variance method is proposed by us. The maximum likelihood estimation method is applicable to any distribution family, while our mean-variance method can estimate the parameters of distributions that have at most two parameters (many widely used distributions have at most two parameters) more efficiently than the maximum likelihood estimation method. 

{\bf Maximum likelihood estimation method.}  For each frequency estimation $\hat{f}_i$, we can compute its probability as $\sum_{k}p_{f}(f=k|\Theta)p_{s}(s$ $=\hat{f}_i-k)$, which is a function of the parameters $\Theta$.  In maximum likelihood estimation, we aim to find the parameters $\Theta$ that maximize the product of the probabilities of the $d$ frequency estimations. Specifically, we obtain $\Theta$ via solving the following optimization problem:
\begin{align}
\label{mle}
\Theta=\argmax_{\theta} \sum_i log(\sum_{k}p_{f}(f=k|\theta)p_{s}(s=\hat{f}_i-k)).
\end{align}
We can use the \emph{gradient descent} method to solve the optimization problem iteratively. 
Specifically,  $\Theta$ is initialized to be some random value; 
in each iteration, we compute the gradient of the objective function with the current $\Theta$, 
and we move $\Theta$ along the gradient with a certain step size. Due to limited space, we omit the details.


{\bf Mean-variance method.} We propose a mean-variance method to estimate the parameters $\Theta$ when $p_{f}$ has at most two parameters. For instance, popular distributions--such as power-law distributions, Gaussian distributions, Laplacian distributions, and Poison distributions--have one or two parameters. For such distributions, our mean-variance method is much more efficient than the maximum likelihood estimation method because solving the optimization problem in Equation~\ref{mle} involves sum over $k$. 
Since $\hat{f} = f + s$, we have:
\begin{align}
\label{exp1}
E(\hat{f}) &= E(f) + E(s) \\
\label{var}
Var(\hat{f}) &= Var(f) + Var(s),
\end{align}
where $E$ and $Var$ represent expectation and variance, respectively.
Given the $d$ frequency estimations $\hat{f}_{1}, \hat{f}_{2}, \cdots, \hat{f}_{d}$, we can estimate $E(\hat{f})$ and $Var(\hat{f})$. Specifically, we have:
\begin{align}
E(\hat{f})&=\frac{1}{d}\sum_i \hat{f}_{i} \\
Var(\hat{f})&=\frac{1}{d}\sum_i (\hat{f}_{i} - E(\hat{f}))^2
\end{align}

Moreover, $E(s)$ and $Var(s)$ are known since $s$ follows a Gaussian distribution, according to our analysis in Section~\ref{estimateps}. $E(f)$ and $Var(f)$ are functions of $\Theta$. Therefore, Equation~\ref{exp1} and~\ref{var} define a system of equations for the parameters $\Theta$. Via solving the system, we obtain the parameters $\Theta$.  When the distribution $p_{f}$ has one parameter, we can simply use Equation~\ref{exp1} to solve the parameter.

%
%

\myparatight{Estimating $p_{\hat{f}}$} Since the three random variables $s$, $f$, and $\hat{f}$ are correlated as $\hat{f} = f + s$, the three probability distributions $p_s$, $p_f$, and $p_{\hat{f}}$ have the following relationship:
\begin{align}
 p_{\hat{f}}(\hat{f} = \hat{k})=\sum_k p_f(f=k)p_s(s=\hat{k}-k).
 \end{align}
Therefore, given the probability distributions $p_s$ and $p_f$, we can estimate $p_{\hat{f}}$.



\section{Evaluation}
\label{exp}


\subsection{Experimental Setup}

\subsubsection{Two Real-world Datasets}  

\begin{figure}[!t]
\vspace{-2mm}
\centering
\subfloat[Kosarak]{\includegraphics[width=0.25 \textwidth]{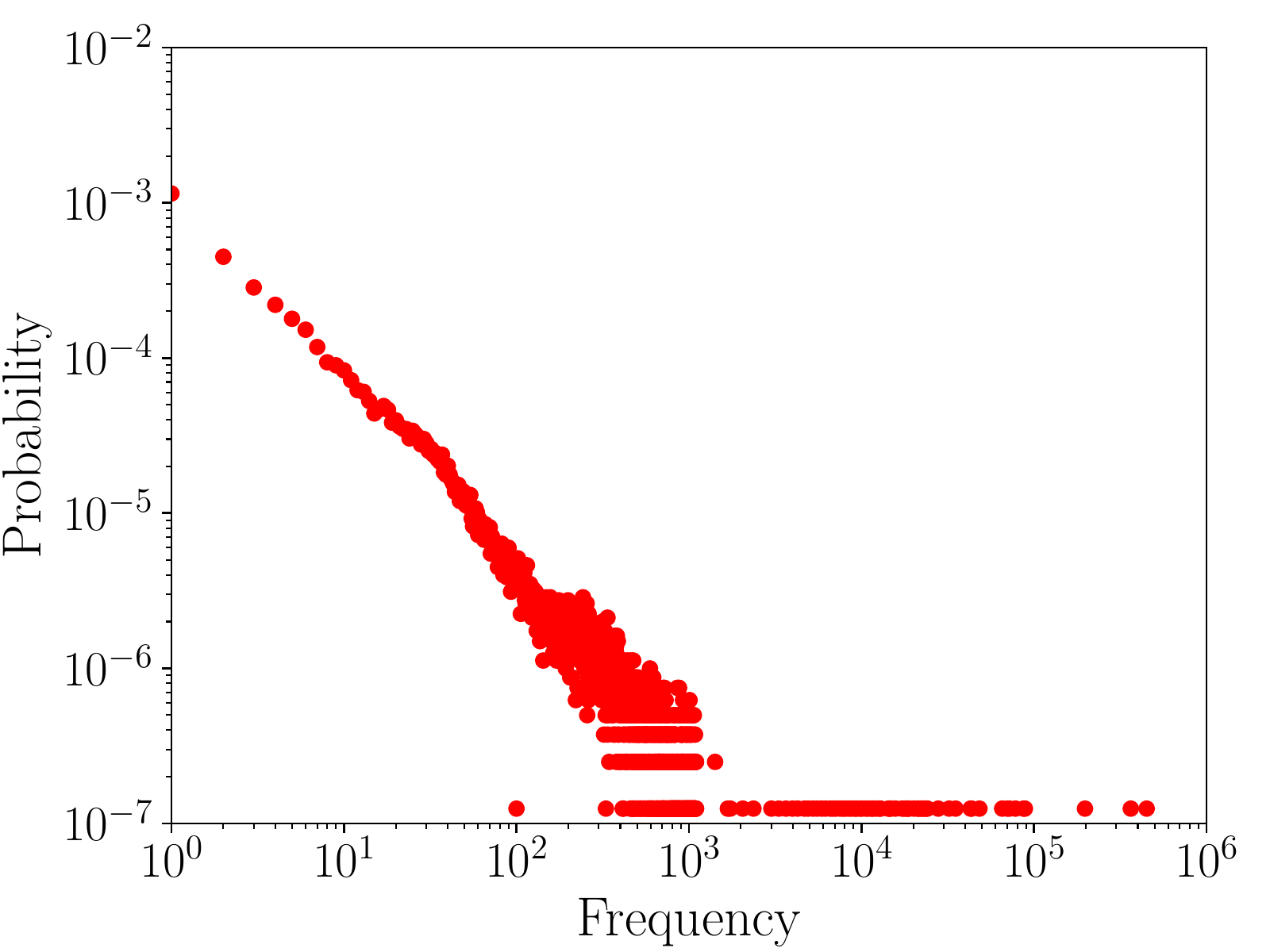}\label{kosarak}}
\subfloat[Retail Market Basket]{\includegraphics[width=0.25 \textwidth]{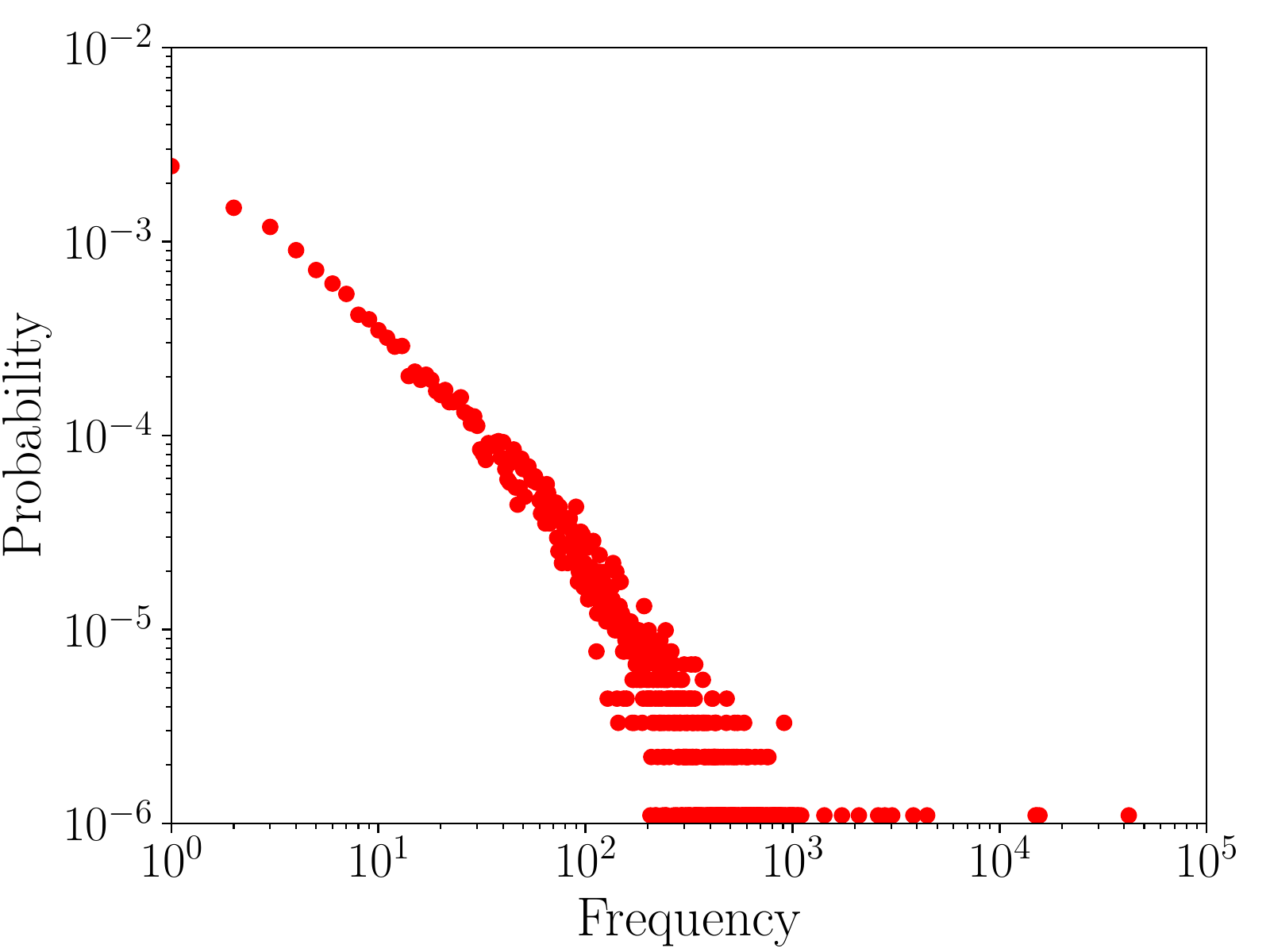}\label{retail}}
\caption{True item frequencies in the (a) Kosarak dataset and (b) Retail Market Basket dataset follow power-law distributions.}
\vspace{-2mm}
\end{figure}

Like previous studies~\cite{QinHeavyHitterCCS16,Ninghui:2017}, we use the {Kosarak}~\cite{Kosarakdata} and Retail Market Basket~\cite{brijs1999using} datasets. In the Kosarak dataset, items are webpages and an item's frequency is the number of user clicks of the corresponding webpage. Overall, there are 41,270 webpages and their total number of frequencies is 8,019,015. In Retail Market Basket dataset, items represent products in supermarket store and the frequency of an item is the sale of the product. In summary, there are 16,470 products and the total frequencies is 908,576. Like a previous study~\cite{Ninghui:2017}, we assume each user generates one webpage click or buys one product, i.e., each user has one item. 

Like many real-world phenomena, the number of clicks (i.e., frequencies) of the webpages or the sale of the product roughly follows a power-law distribution. Specifically, Figure~\ref{kosarak} and~\ref{retail} respectively show the probability distribution of the item frequencies in the two datasets, where the x-axis is item frequency and y-axis is the fraction of items that have a given frequency. The curve is a typical power-law distribution observed in real world under log-log scale~\cite{cha2007tube,Clauset09}: a large fraction of items have small frequencies, a very small fraction of items have large frequencies (i.e., the known \emph{long-tail} property), and the curve is close to a line when the item frequency is smaller than a certain threshold. 


\begin{figure}[!t]
\centering
{\includegraphics[width=0.35 \textwidth]{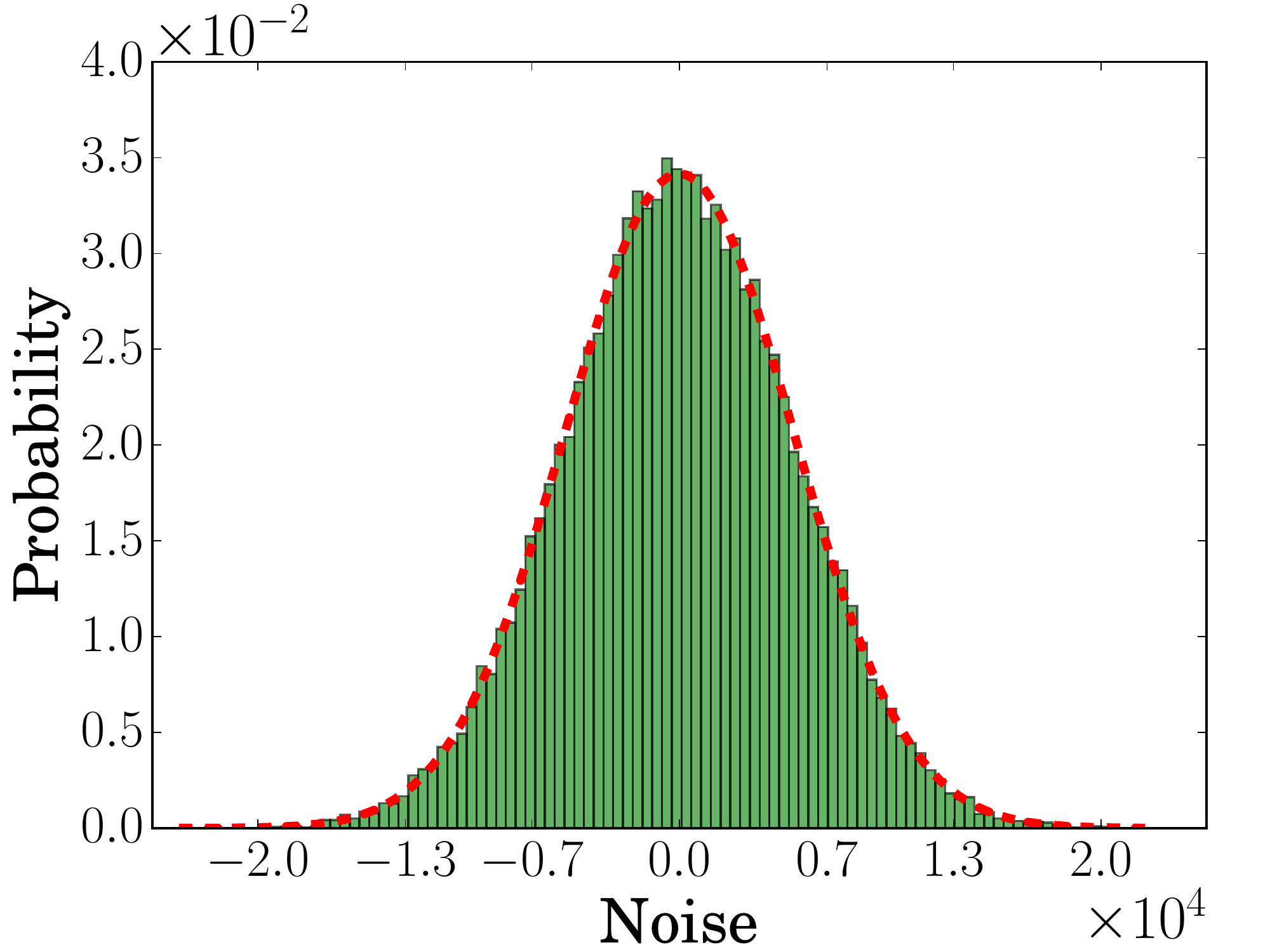}}
\caption{Empirical noise distribution.}
\label{noisedistribution}
\vspace{-4mm}
\end{figure}

\begin{figure*}[!t]
\centering
\subfloat[Precision]{\includegraphics[width=0.33 \textwidth]{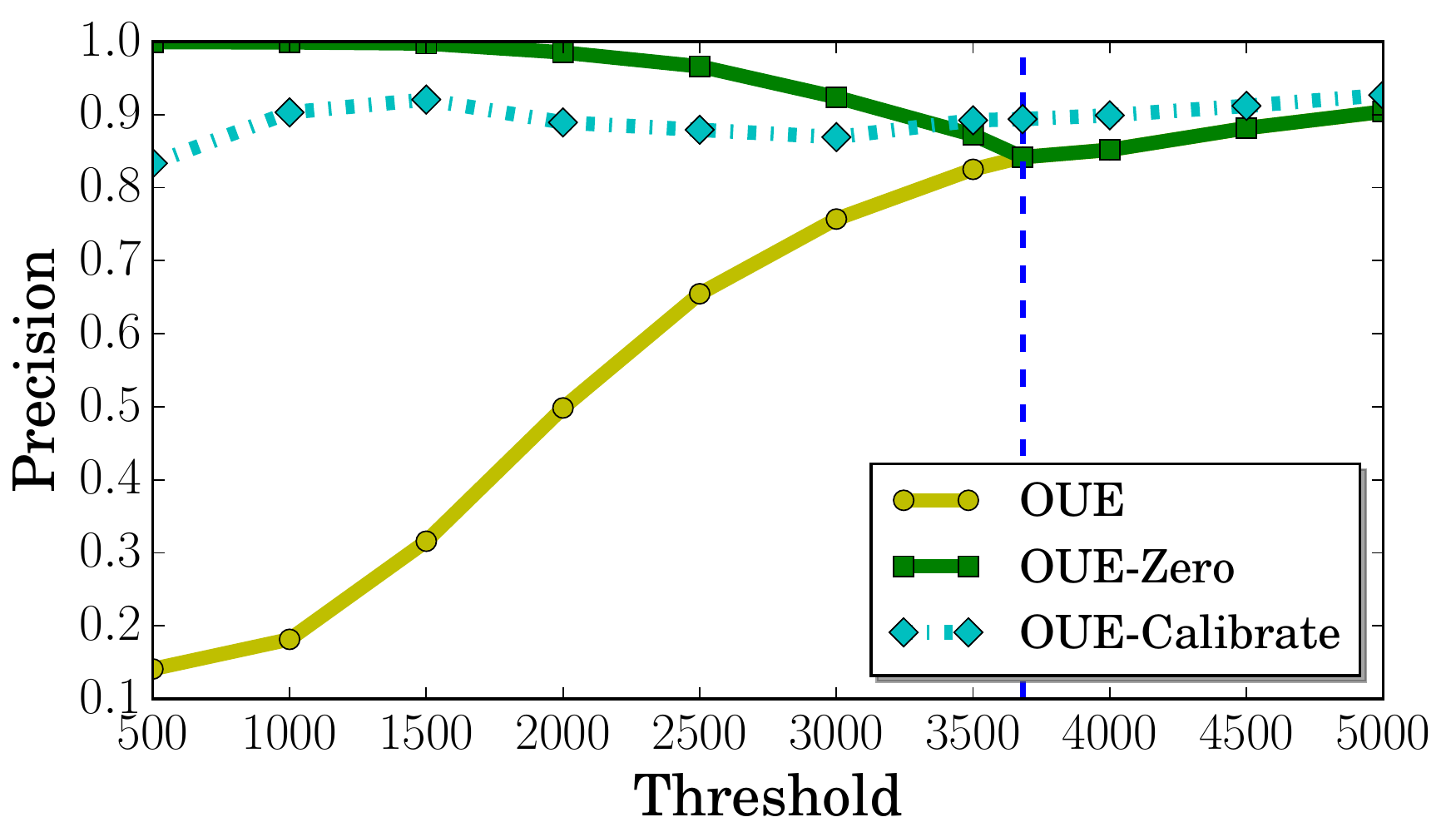}}
\subfloat[Recall]{\includegraphics[width=0.33 \textwidth]{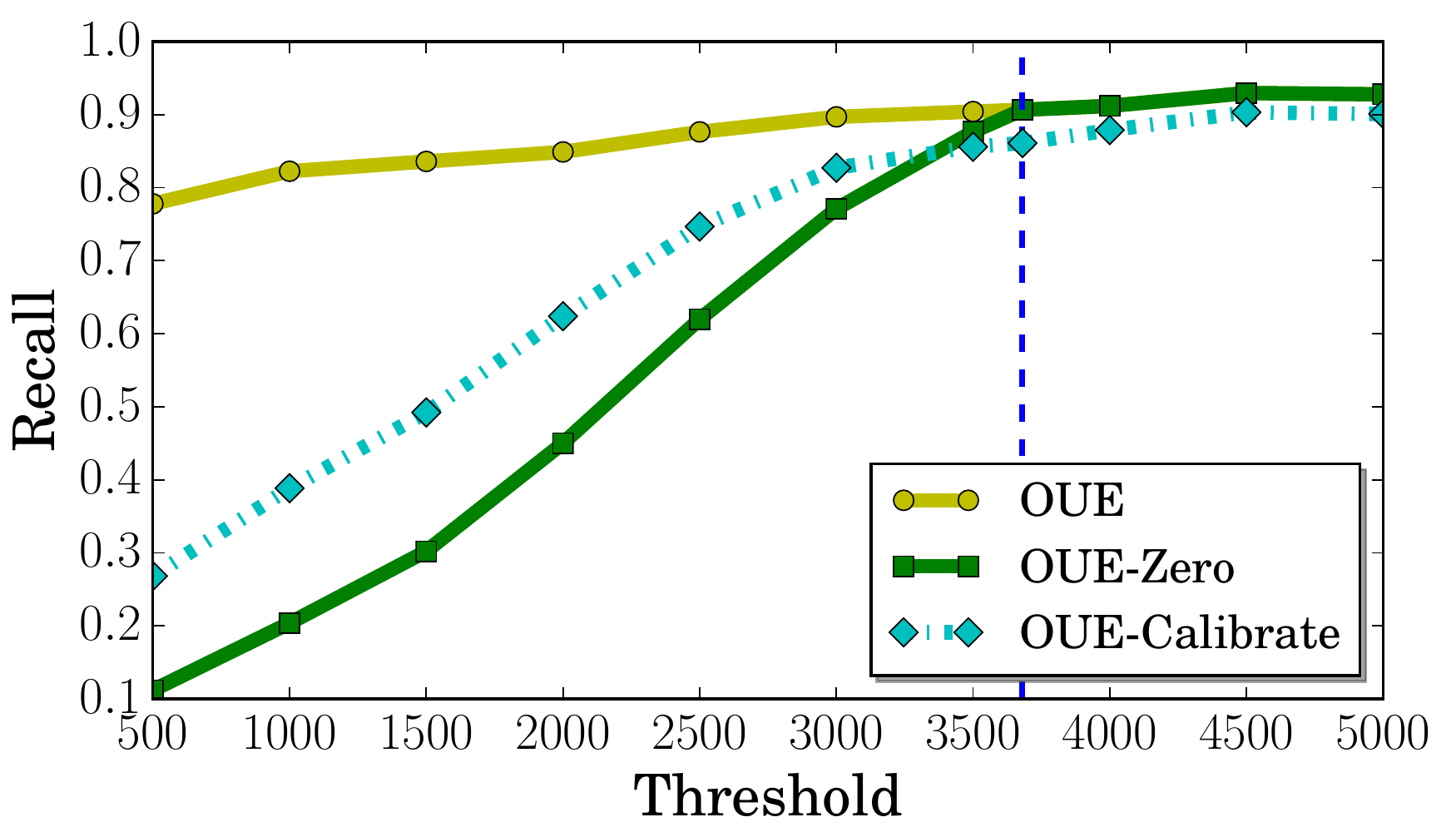}}
\subfloat[F-Score]{\includegraphics[width=0.33 \textwidth]{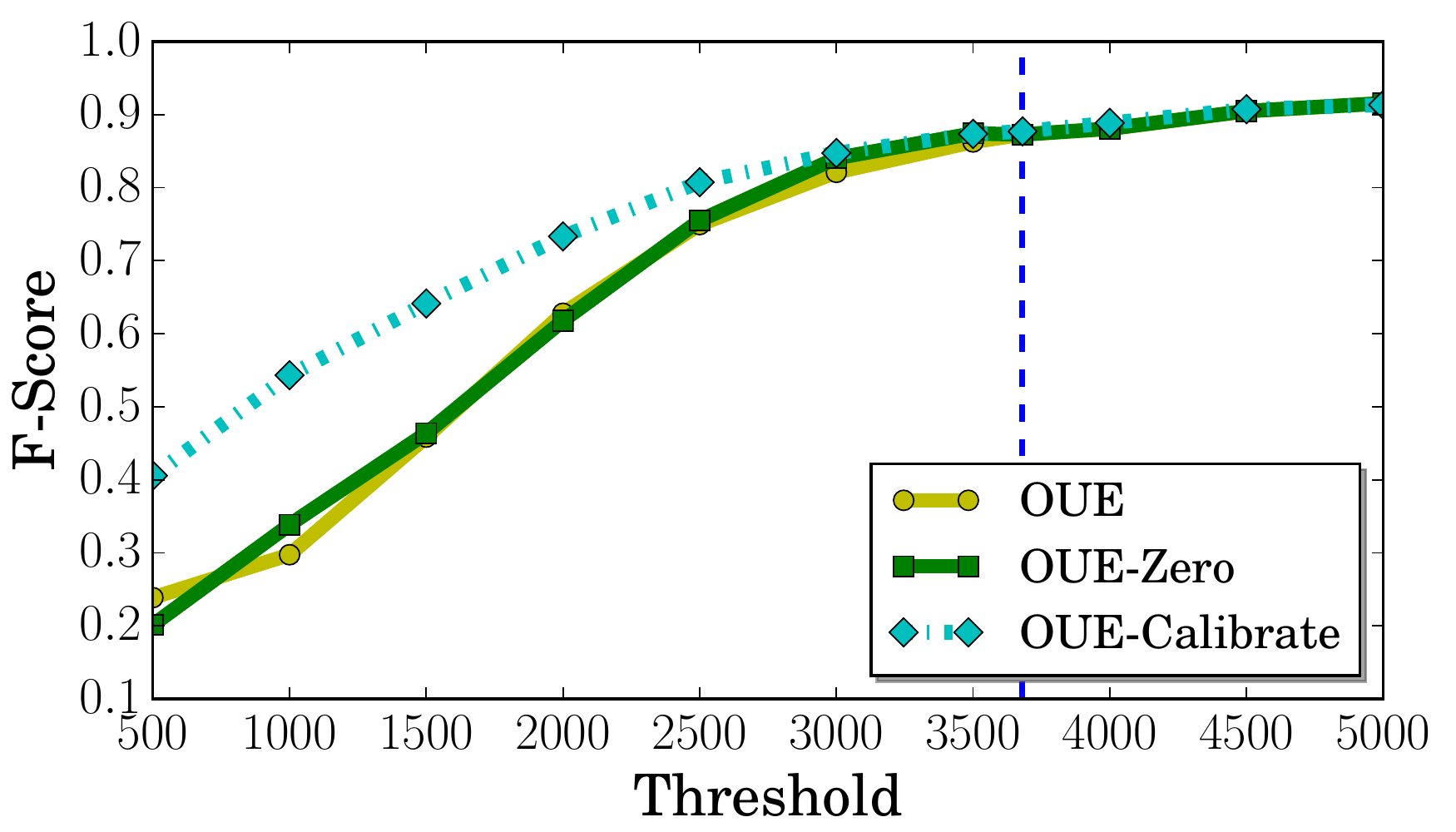}}

\caption{Precision, Recall, and F-Score of heavy hitter identification vs. threshold in Kosarak. The vertical line is the significance threshold.}
\label{estimationkosarak}
\end{figure*}

\begin{figure*}[!t]
\centering
\subfloat[Precision]{\includegraphics[width=0.33 \textwidth]{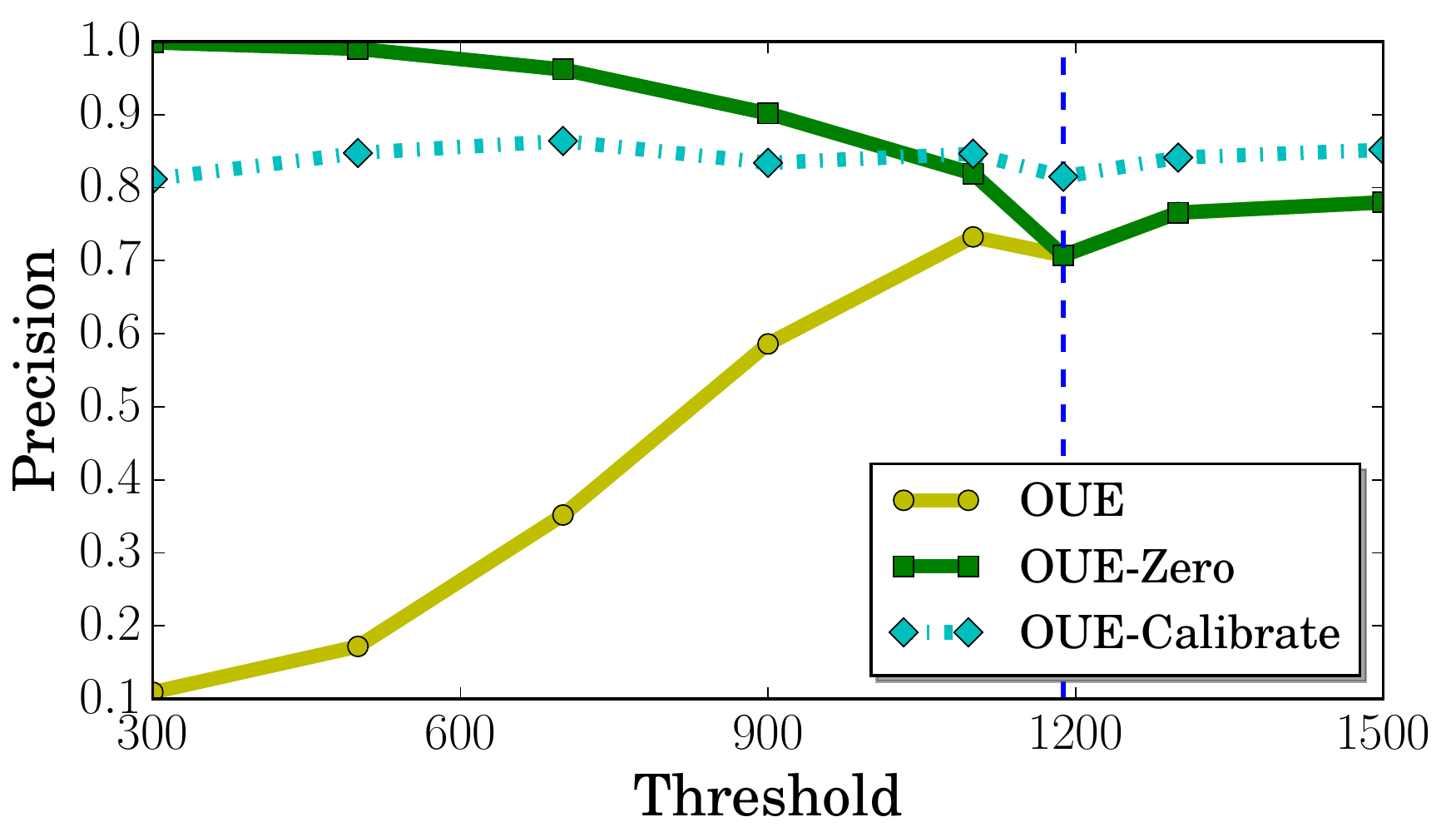}}
\subfloat[Recall]{\includegraphics[width=0.33 \textwidth]{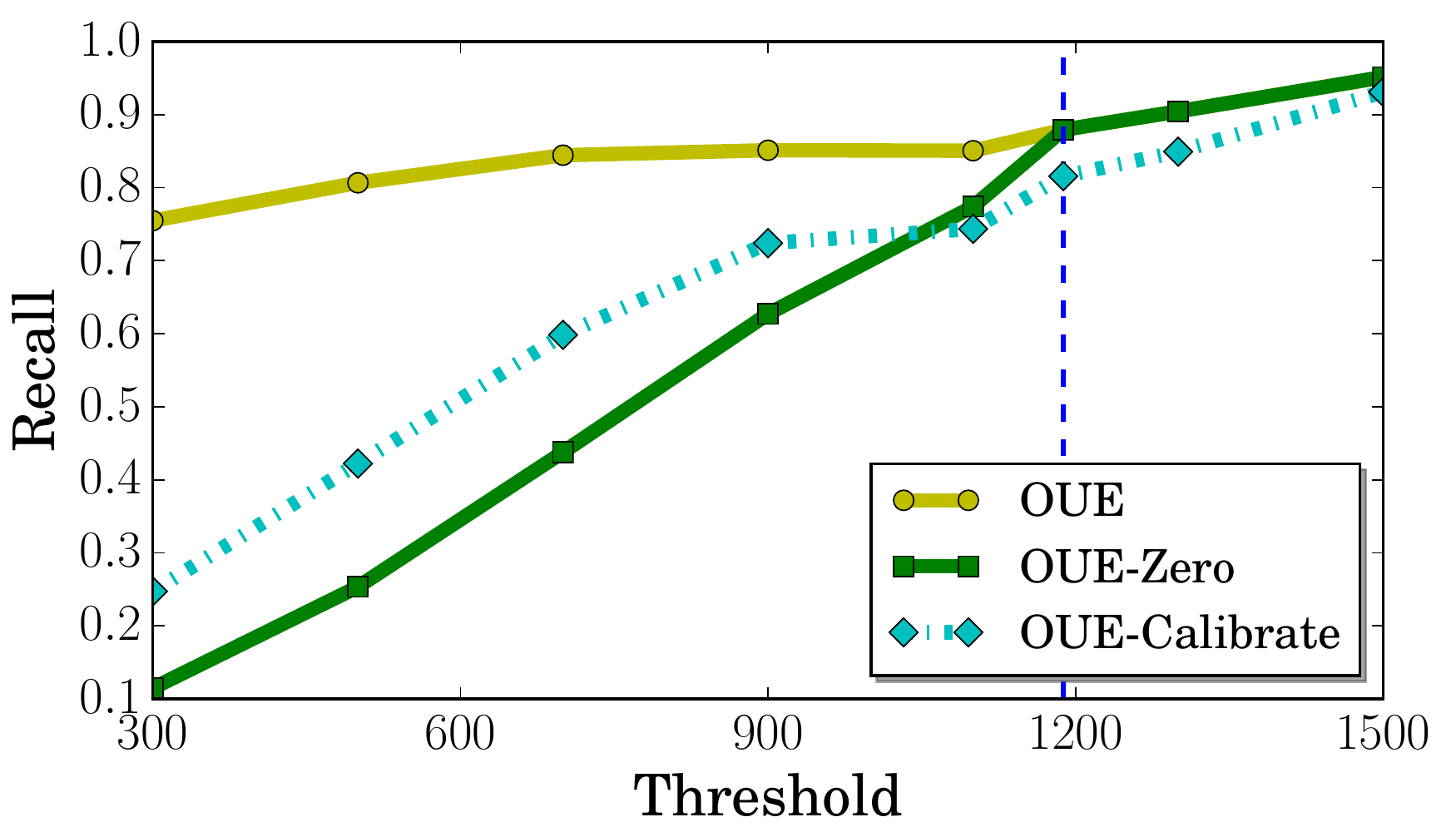}}
\subfloat[F-Score]{\includegraphics[width=0.33 \textwidth]{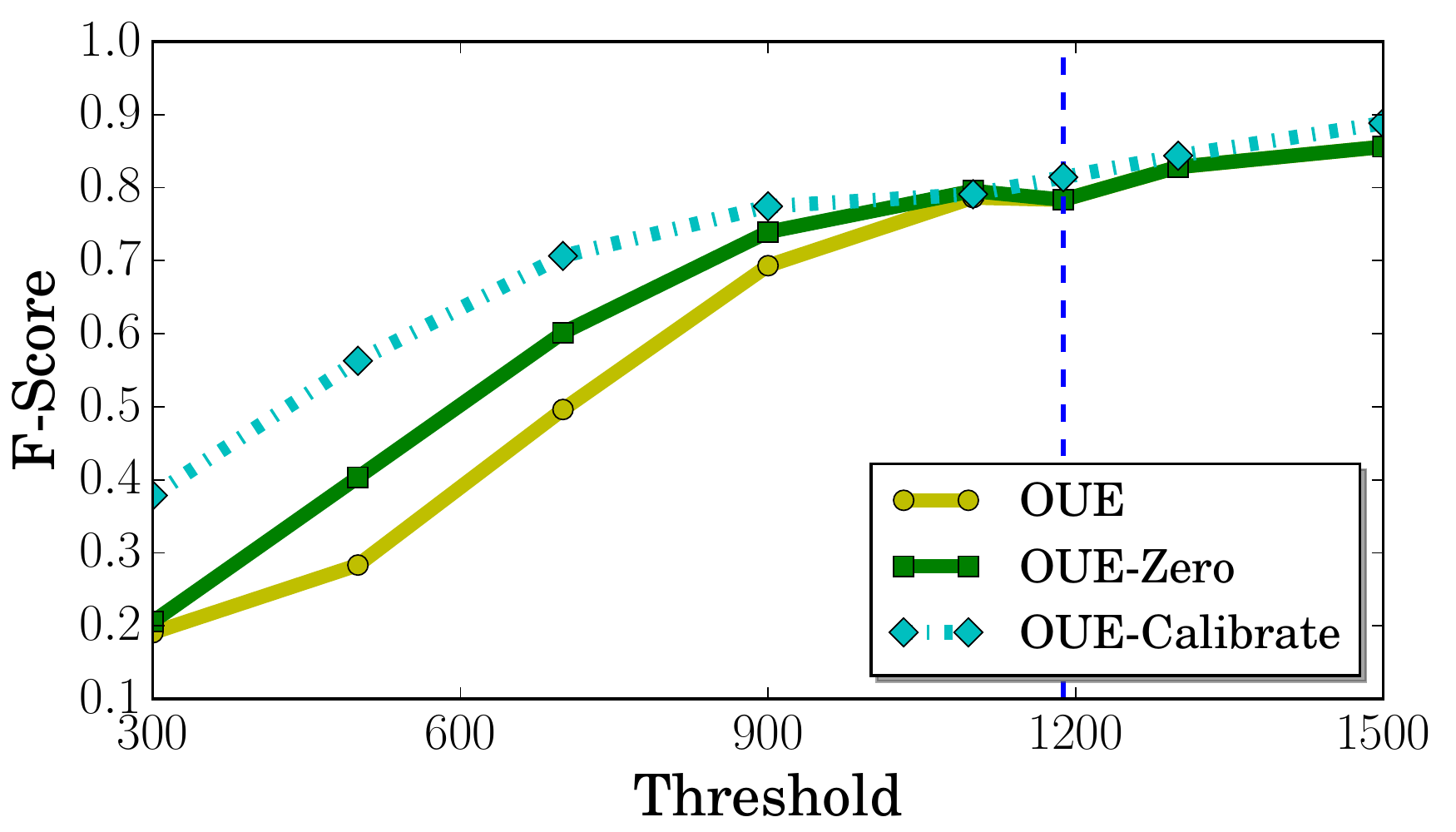}}

\caption{Precision, Recall, and F-Score of heavy hitter identification vs. threshold in  Retail Market Basket. The vertical line is the significance threshold.}
\label{estimationretail}
\end{figure*}

\subsubsection{Compared Methods} We compare the following LDP algorithms:


\myparatight{Optimized Unary Encoding (OUE)~\cite{Ninghui:2017}} OUE is an optimized version of basic RAPPOR~\cite{Erlingsson:2014} and achieves state-of-the-art performance. Please refer to Section~\ref{relatedwork} for the $Encode$, $Perturb$, and $Aggregate$ steps of OUE.  

\myparatight{OUE-Zero} The high estimated item frequencies are statistically more reliable, as the probability that they are generated from low frequencies is low. However, the estimated low item frequencies may be less reliable because of the noise in the LDP process. In existing methods~\cite{Erlingsson:2014,Ninghui:2017}, frequencies that are smaller than a \emph{significance threshold} are unreliable and a data collector can discard them. Specifically, the significance threshold is defined as follows:
\begin{align}
\text{\bf Significance threshold = } \phi^{-1}(1-\frac{\beta}{d})\sqrt{Var},
 \end{align}
where $\phi^{-1}$ is the inverse of the
cumulative density function of standard Gaussian distribution, $d$ is the number of items, $Var$ is the estimation variance of the method, and we set $\beta=0.05$ as suggested by prior work~\cite{Ninghui:2017}. Therefore, we also evaluate OUE-Zero, which sets the estimated frequencies that are smaller than the significance threshold to be zero.

\begin{figure}[t]
\vspace{-2mm}
\centering
\subfloat[Kosarak]{\includegraphics[width=0.25 \textwidth]{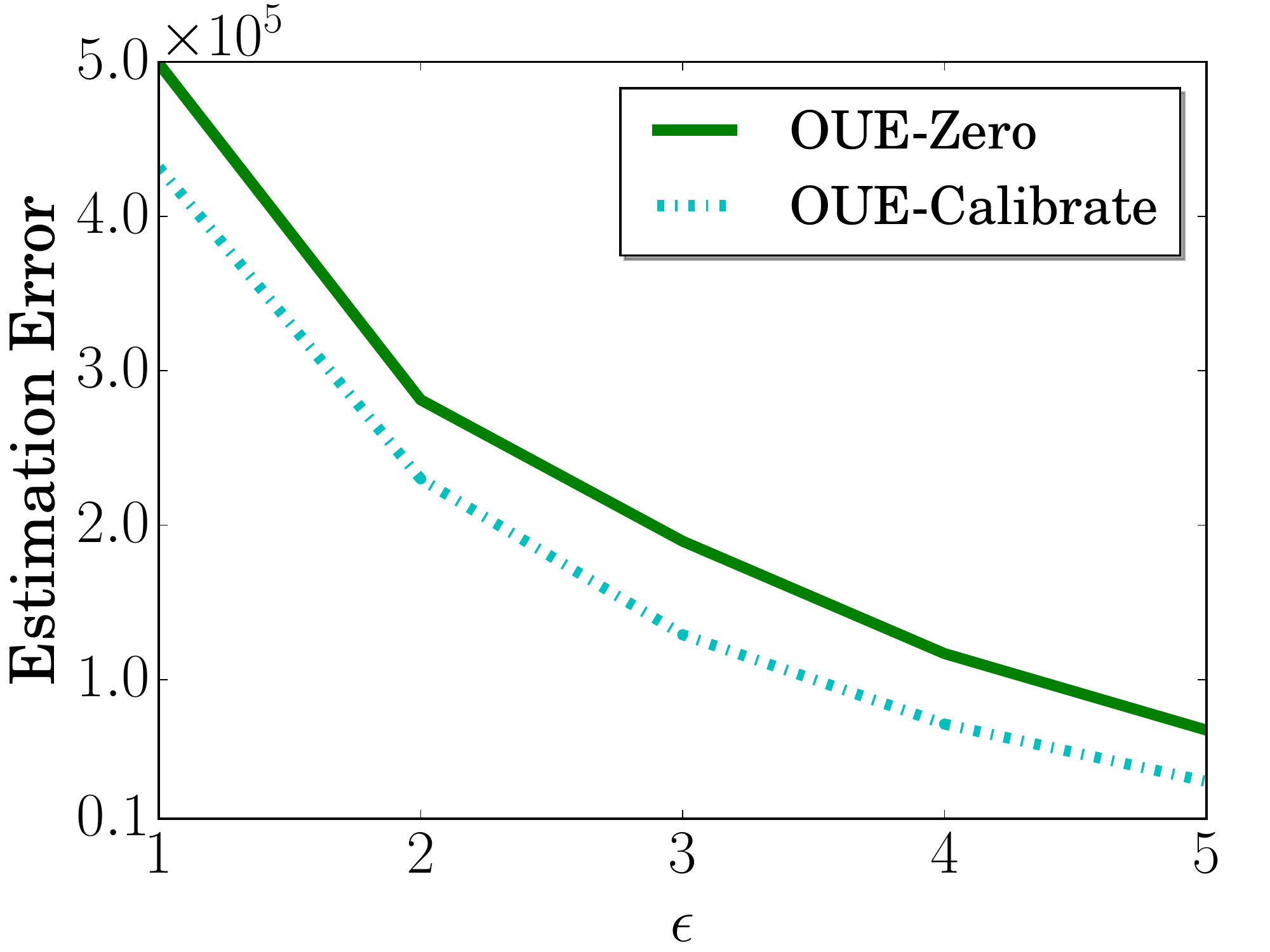}}
\subfloat[Retail Market Basket]{\includegraphics[width=0.25 \textwidth]{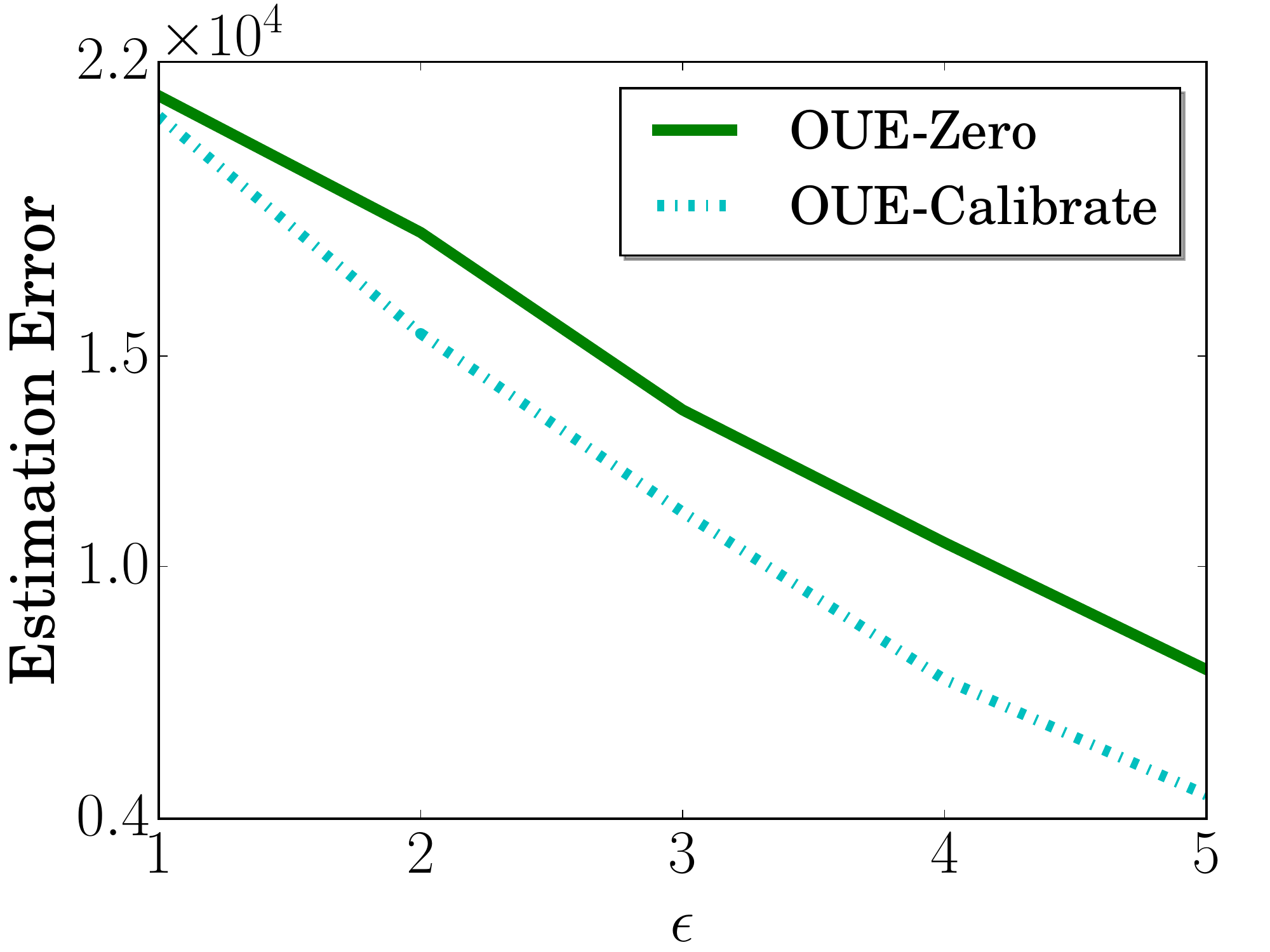}}
\caption{Estimation errors for frequency estimation on the two datasets as a function of the privacy budget $\epsilon$. Both OUE-Zero and OUE-$Calibrate$ have orders of magnitude of smaller estimation errors than OUE. Therefore,
to better contrast the difference between OUE-Zero and OUE-$Calibrate$, we omit the results of OUE.}
\label{privacybudget}
\end{figure}

\myparatight{OUE-Calibrate} We append $Calibrate$ to OUE.

%
%
We obtained the publicly available implementation of OUE from its authors~\cite{Ninghui:2017}, and we implemented our $Calibrate$ step in Python. For each experimental setting, we repeat the experiments for dozens of times and compute the average performance, e.g., estimation error for frequency estimation and F-Score for heavy hitter identification.

\begin{figure*}[!t]
\vspace{-2mm}
\centering
\subfloat[Precision]{\includegraphics[width=0.33 \textwidth]{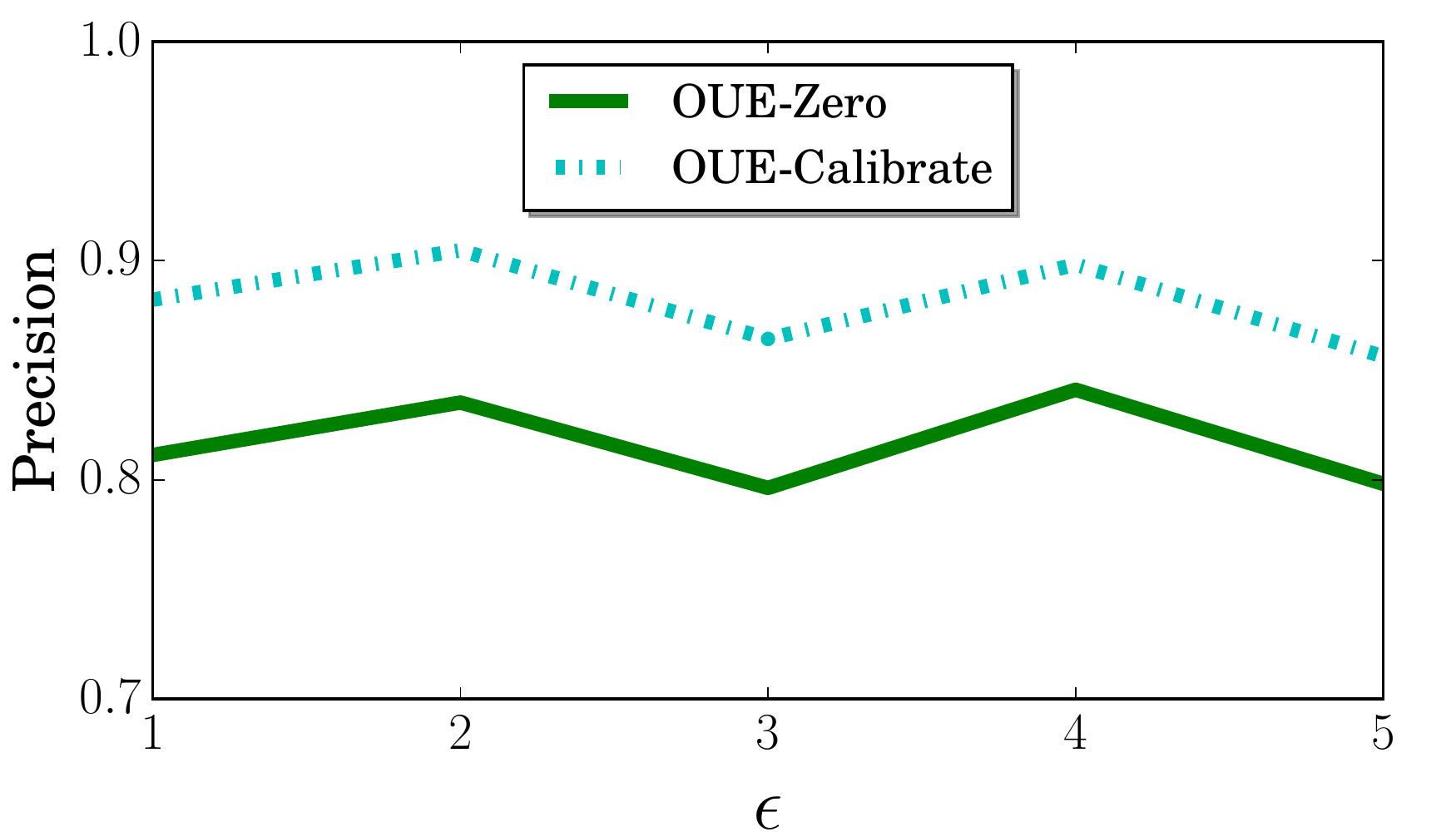}}
\subfloat[Recall]{\includegraphics[width=0.33 \textwidth]{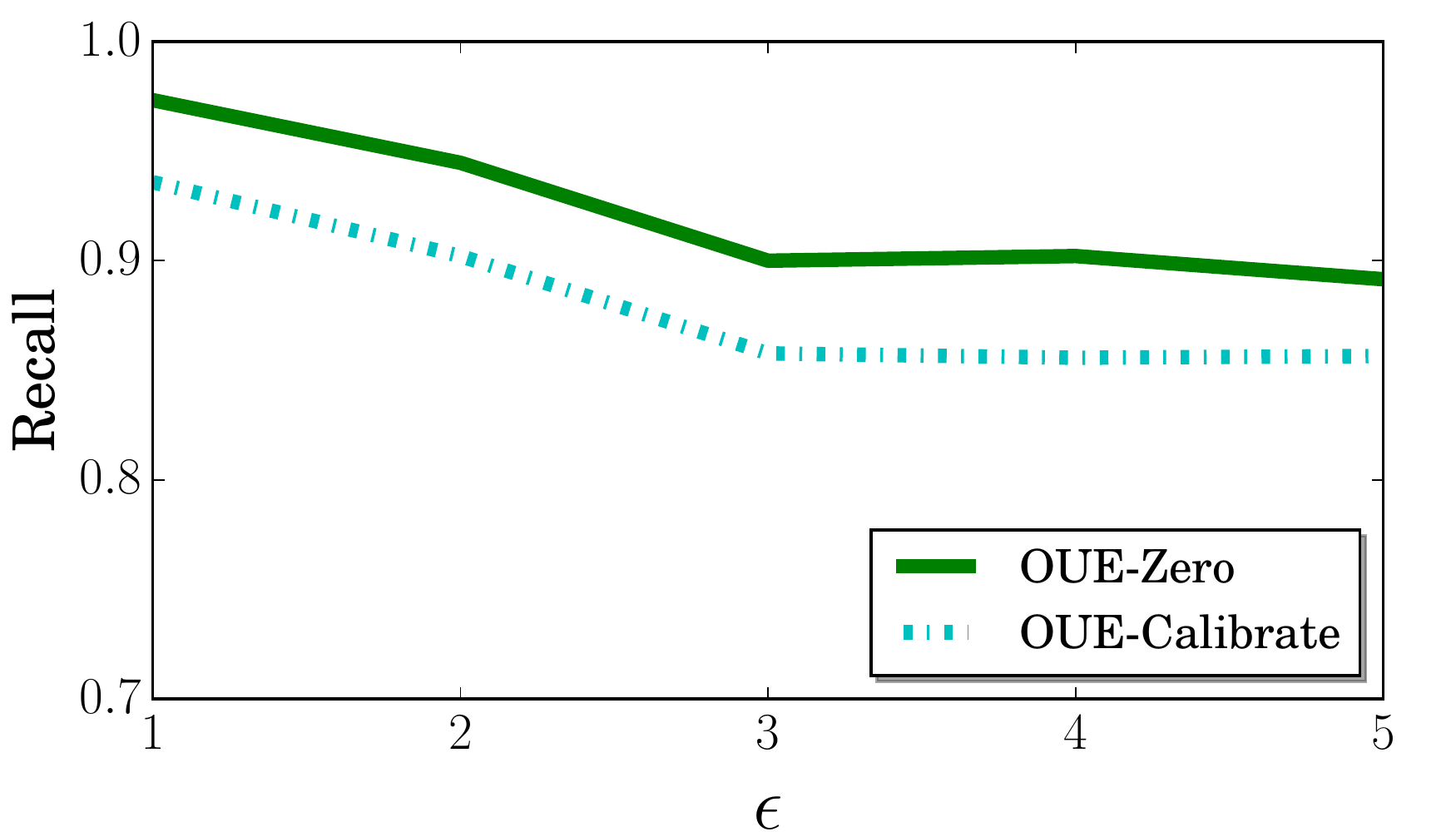}}
\subfloat[F-Score]{\includegraphics[width=0.33 \textwidth]{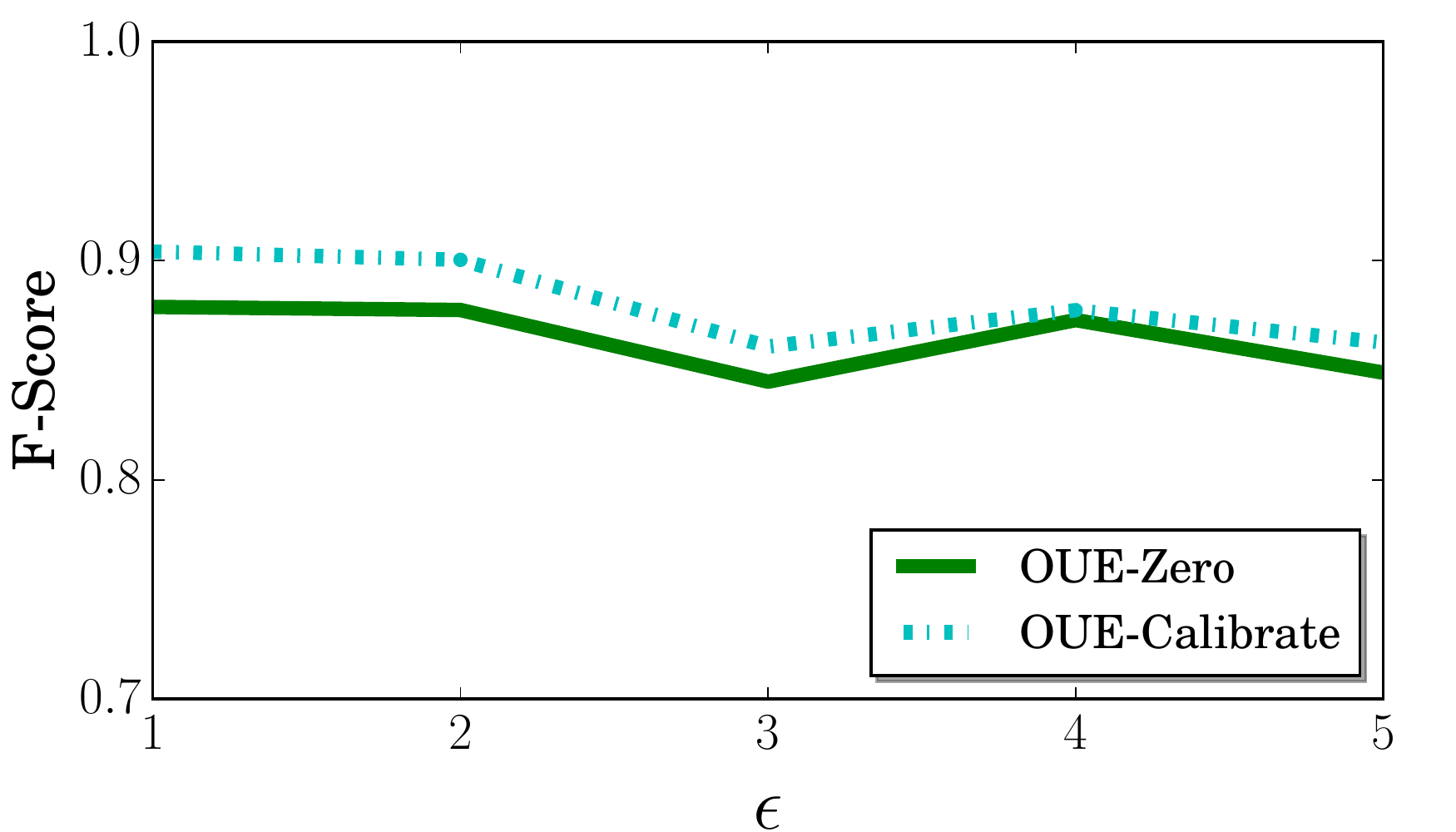}}

\caption{Precision, Recall, and F-Score of heavy hitter identification vs. privacy budget in Kosarak.}
\vspace{-4mm}
\label{ekosarak}
\end{figure*}

\begin{figure*}[!t]
\centering
\subfloat[Precision]{\includegraphics[width=0.33 \textwidth]{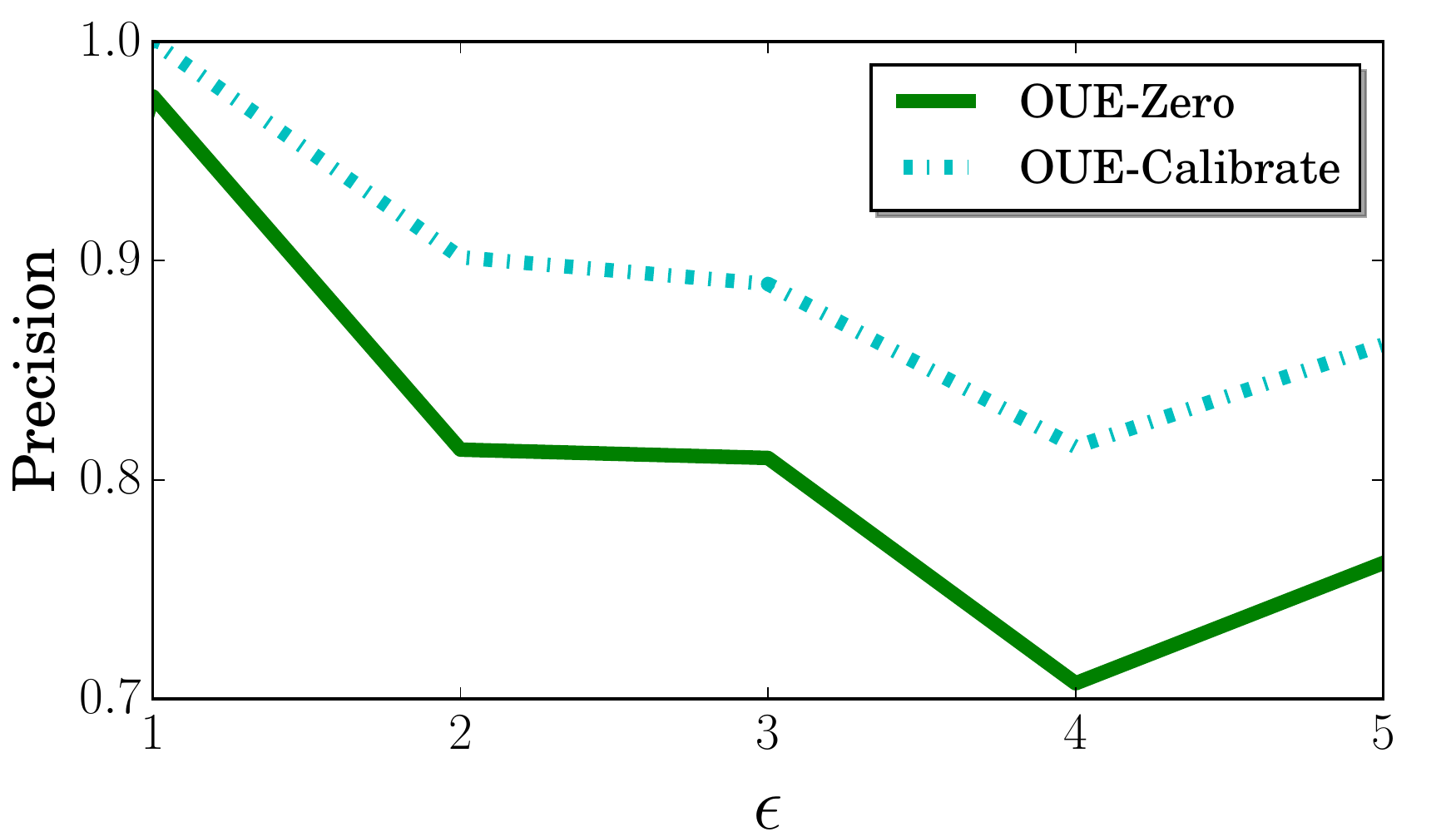}}
\subfloat[Recall]{\includegraphics[width=0.33 \textwidth]{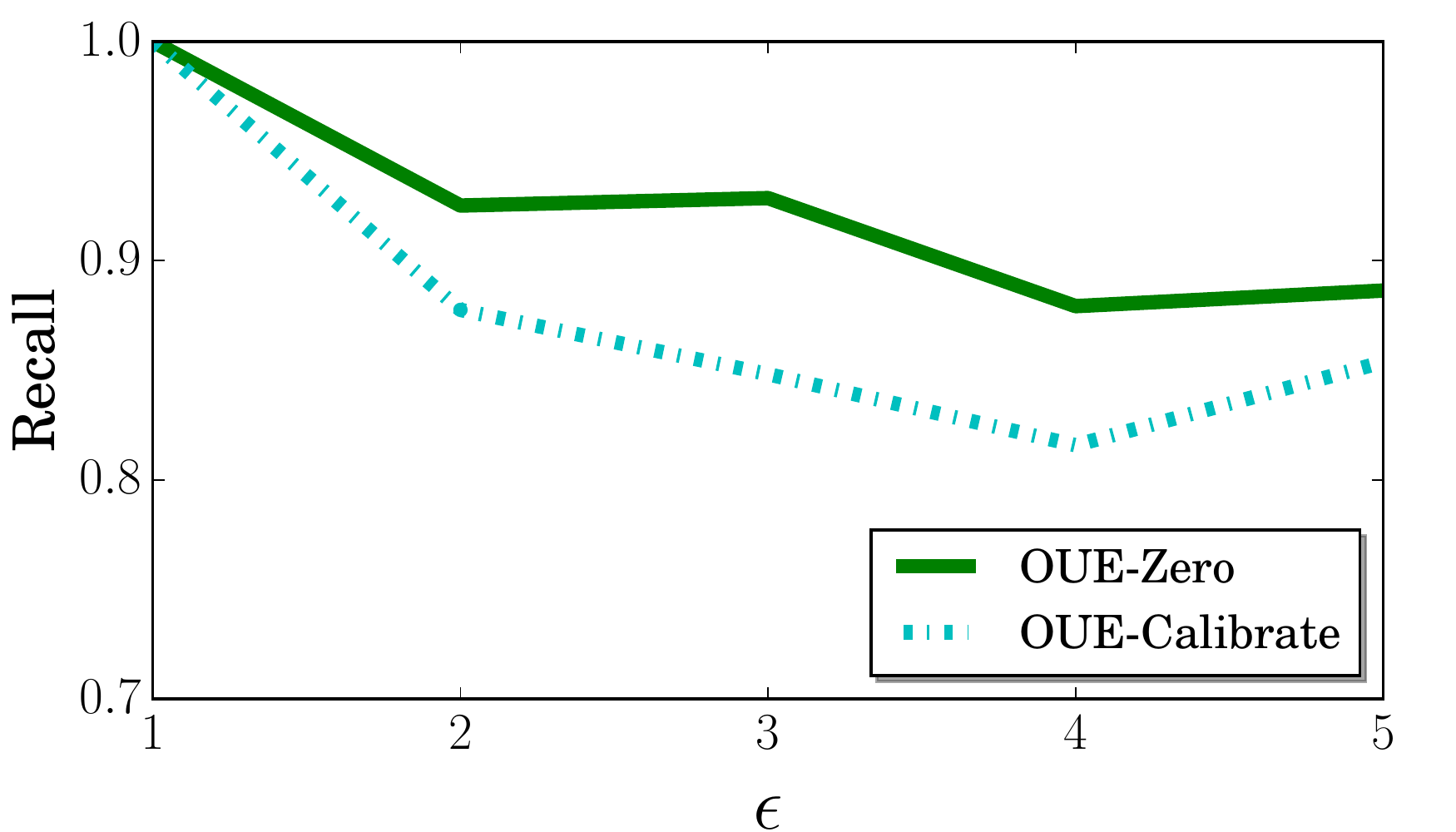}}
\subfloat[F-Score]{\includegraphics[width=0.33 \textwidth]{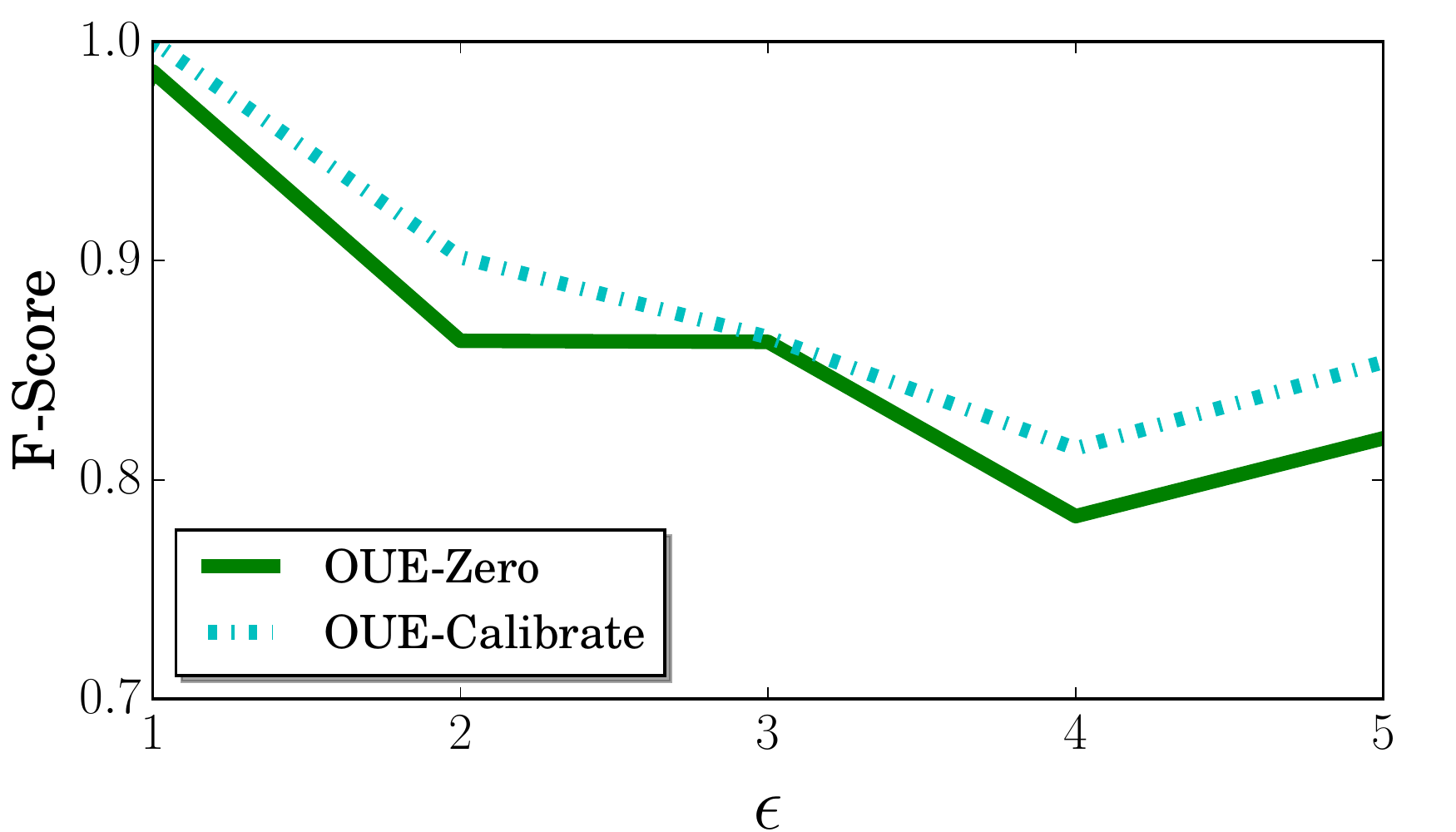}}

\caption{Precision, Recall, and F-Score of heavy hitter identification vs. privacy budget in  Retail Market Basket.}
\vspace{-4mm}
\label{eretail}
\end{figure*}

\subsection{Results}

\myparatight{Noise distribution $p_s$ is a Gaussian distribution} Figure~\ref{noisedistribution} shows an empirical noise distribution for OUE on the Kosarak dataset. Based on the  analysis in Section~\ref{estimateps}, we can theoretically compute that the noise distribution has a mean of 0 and a variance of 5434.  We fit the empirical noise distribution with a Gaussian distribution. The fitted Gaussian distribution has a mean of -3.2 and a variance of 5405, which are very close to the theoretically predicted mean and variance, respectively.

\myparatight{$Calibrate$ reduces estimation errors for frequency estimation} Figure~\ref{privacybudget} shows the estimation errors of OUE-Zero and OUE-$Calibrate$ on the two real-world datasets. We find that both OUE-Zero and OUE-$Calibrate$ outperform OUE by orders of magnitude, and thus we omit the results of OUE to better illustrate the difference between OUE-Zero and OUE-$Calibrate$. OUE-Zero outperforms OUE because OUE-Zero resets the frequencies that OUE is unconfident to be 0, while OUE-$Calibrate$ incorporates prior knowledge to recalibrate these frequencies. We also find that further resetting the frequencies smaller than the significance threshold in OUE-$Calibrate$ has a negligible impact on the estimation errors.  

OUE-$Calibrate$ outperforms OUE-Zero. This is consistent with our Theorem~\ref{optimality}, which shows that $Calibrate$ is the optimal post-processing method. Moreover, the performance gain is more significant as the privacy budget $\epsilon$ increases. For instance, on the Kosarak dataset, the relative improvements of OUE-$Calibrate$ upon OUE-Zero are 16\% and 97\% for $\epsilon=1$ and $\epsilon=5$, respectively. Likewise, on the Retail Market Basket dataset, the relative improvements of OUE-$Calibrate$ upon OUE-Zero are 2.4\% and 65\% for $\epsilon=1$ and $\epsilon=5$, respectively. 
This is because a larger privacy budget means smaller noise in the estimated frequencies of OUE. Subsequently,  $Calibrate$ can better estimate the probability distributions that model the prior knowledge from the estimated frequencies. As a result,  OUE-$Calibrate$ shows more improvements over  OUE-Zero.

\myparatight{$Calibrate$ improves heavy hitter identification} 
Figure~\ref{estimationkosarak} and~\ref{estimationretail} show the Precision, Recall, and F-Score for different thresholds on the Kosarak and Retail Market Basket datasets respectively, where $\epsilon=4$. Moreover, Figure~\ref{ekosarak} and~\ref{eretail} show the results for different privacy budgets, where the threshold is chosen as the corresponding significance threshold for each privacy budget.  We select the significance threshold because the frequencies smaller than the significance threshold are unreliable in OUE. OUE and OUE-Zero have the exactly same curves on these graphs, and thus we omit OUE for simplicity. 

Overall, OUE-$Calibrate$ outperforms OUE and OUE-Zero, i.e., OUE-$Calibrate$ has a better F-Score than OUE and OUE-Zero for different thresholds (especially thresholds that are smaller than the significance threshold) and privacy budgets. 
More specifically,  OUE-$Calibrate$ achieves higher Precision but lower Recall than OUE. The prior knowledge about the true item frequencies is that they follow a power-law distribution, which indicates that only a small fraction of items have high frequencies. Therefore, OUE-$Calibrate$ tends to reduce item frequencies of the high item frequencies estimated by OUE, via considering the prior knowledge. As a result, given a threshold, OUE-$Calibrate$ has a smaller number of True Positives and False Positives than OUE, which leads to higher Precision and lower Recall. The improvement of OUE-$Calibrate$ is more significant when the threshold is smaller. This is because OUE has unreliable estimations for the smaller frequencies, OUE-Zero simply resets them to be zero, while OUE-$Calibrate$ uses prior knowledge to recalibrate them. 

When the threshold is larger than the significance threshold, OUE-Zero and OUE have the same performance. 
When the threshold is smaller than the significance threshold, OUE-Zero has higher Precision but lower Recall when the threshold is smaller. This is because OUE-Zero resets all frequencies smaller than the significance threshold to be zero. Therefore, OUE-Zero predicts the same set of items as heavy hitters for all thresholds that are smaller than the significance threshold. As the threshold becomes smaller, the True Positives would increase and False Positives would decrease, which explains the increasing Precision. However, False Negatives also increase and they increase faster than True Positives, which explains the decreasing Recall of OUE-Zero.

\section{Conclusion and Future Work}
Frequency estimation with local differential privacy (LDP) is a basic step in privacy-preserving data analytics without a trusted data collector. 
In this work, we propose $Calibrate$ to calibrate item frequencies estimated by an existing LDP algorithm. Our $Calibrate$ incorporates prior knowledge about noise in the estimated item frequencies and prior knowledge about true item frequencies through statistical inference. We show that such prior knowledge can be modeled as two probability distributions, respectively; and the two probability distributions can be estimated via integrating techniques from statistics and signal processing.  Our empirical results on both synthetic and real-world datasets demonstrate that  our $Calibrate$ can reduce estimation errors of state-of-the-art LDP algorithms by orders of magnitude. An interesting future work is to generalize our $Calibrate$ to calibrate results for other data analytics tasks such as frequent pattern mining. 

\noindent
{\bf Acknowledgements:} We thank the anonymous reviewers for their insightful reviews. This work was supported by NSF grant No. 1801584.

{
\balance{
\bibliographystyle{IEEEtranS}
\bibliography{refs}

\begin{thebibliography}{10}
\providecommand{\url}[1]{#1}
\csname url@samestyle\endcsname
\providecommand{\newblock}{\relax}
\providecommand{\bibinfo}[2]{#2}
\providecommand{\BIBentrySTDinterwordspacing}{\spaceskip=0pt\relax}
\providecommand{\BIBentryALTinterwordstretchfactor}{4}
\providecommand{\BIBentryALTinterwordspacing}{\spaceskip=\fontdimen2\font plus
\BIBentryALTinterwordstretchfactor\fontdimen3\font minus
  \fontdimen4\font\relax}
\providecommand{\BIBforeignlanguage}[2]{{%
\expandafter\ifx\csname l@#1\endcsname\relax
\typeout{** WARNING: IEEEtranS.bst: No hyphenation pattern has been}%
\typeout{** loaded for the language `#1'. Using the pattern for}%
\typeout{** the default language instead.}%
\else
\language=\csname l@#1\endcsname
\fi
#2}}
\providecommand{\BIBdecl}{\relax}
\BIBdecl

\bibitem{acharya2018hadamard}
J.~Acharya, Z.~Sun, and H.~Zhang, ``Hadamard response: Estimating distributions
  privately, efficiently, and with little communication,'' \emph{arXiv preprint
  arXiv:1802.04705}, 2018.

\bibitem{a2009height}
B.~A'hearn, F.~Peracchi, and G.~Vecchi, ``Height and the normal distribution:
  evidence from italian military data,'' \emph{Demography}, 2009.

\bibitem{AppLDP17}
{Apple Differential Privacy Team}, ``Learning with privacy at scale,'' in
  \emph{Machine Learning Journal}, 2017.

\bibitem{Blender:2017}
B.~Avent, A.~Korolova, D.~Zeber, T.~Hovden, and B.~Livshits, ``Blender:
  Enabling local search with a hybrid differential privacy model,'' in
  \emph{Usenix Security Symposium}, 2017.

\bibitem{BassilySuccinctHistograms15}
R.~Bassily and A.~D. Smith, ``Local, private, efficient protocols for succinct
  histograms,'' in \emph{STOC}, 2015.

\bibitem{bassily17pratical}
R.~Bassily, K.~Nissim, U.~Stemmer, and A.~Thakurta, ``Practical locally private
  heavy hitters,'' in \emph{NIPS}, 2017.

\bibitem{brijs1999using}
T.~Brijs, G.~Swinnen, K.~Vanhoof, and G.~Wets, ``Using association rules for
  product assortment decisions: A case study,'' in \emph{KDD}, 1999.

\bibitem{bun2018heavy}
M.~Bun, J.~Nelson, and U.~Stemmer, ``Heavy hitters and the structure of local
  privacy,'' in \emph{SIGMOD}.\hskip 1em plus 0.5em minus 0.4em\relax ACM,
  2018.

\bibitem{cha2007tube}
M.~Cha, H.~Kwak, P.~Rodriguez, Y.-Y. Ahn, and S.~Moon, ``I tube, you tube,
  everybody tubes: analyzing the world's largest user generated content video
  system,'' in \emph{ACM Internet Measurement Conference}, 2007.

\bibitem{Clauset09}
A.~Clauset, C.~R. Shalizi, and M.~E.~J. Newman, ``Power-law distributions in
  empirical data,'' \emph{SIAM Review}, no.~51, 2009.

\bibitem{degroot2011probability}
M.~H. DeGroot and M.~J. Schervish, \emph{Probability and Statistics (4th
  Edition)}, 2011.

\bibitem{ding2017collecting}
B.~Ding, J.~Kulkarni, and S.~Yekhanin, ``Collecting telemetry data privately,''
  in \emph{NIPS}, 2017.

\bibitem{duchi2013local}
J.~C. Duchi, M.~I. Jordan, and M.~J. Wainwright, ``Local privacy and
  statistical minimax rates,'' in \emph{FOCS}, 2013.

\bibitem{duchi2013localb}
J.~C. Duchi, M.~J. Wainwright, and M.~I. Jordan, ``Local privacy and minimax
  bounds: Sharp rates for probability estimation,'' in \emph{NIPS}, 2013.

\bibitem{Dwork:2006}
C.~Dwork, F.~McSherry, K.~Nissim, and A.~Smith, ``Calibrating noise to
  sensitivity in private data analysis,'' in \emph{TCC}, 2006.

\bibitem{equifaxdataleak}
{Equifax Announces Cybersecurity Incident Involving Consumer Information},
  ``\url{https://goo.gl/HgPXek}.''

\bibitem{gaboardi2017local}
M.~Gaboardi and R.~Rogers, ``Local private hypothesis testing: Chi-square
  tests,'' \emph{arXiv preprint arXiv:1709.07155}, 2017.

\bibitem{Gong12-imc}
N.~Z. Gong, W.~Xu, L.~Huang, P.~Mittal, E.~Stefanov, V.~Sekar, and D.~Song,
  ``Evolution of social-attribute networks: Measurements, modeling, and
  implications using google+,'' in \emph{IMC}, 2012.

\bibitem{Hsu12}
J.~Hsu, S.~Khanna, and A.~Roth, ``Distributed private heavy hitters,'' in
  \emph{Automata, Languages, and Programming}, 2012.

\bibitem{kairouz2016discrete}
P.~Kairouz, K.~Bonawitz, and D.~Ramage, ``Discrete distribution estimation
  under local privacy,'' \emph{ICML}, 2016.

\bibitem{Kosarakdata}
Kosarak, ``\url{http://fimi.ua.ac.be/data/}.''

\bibitem{QinHeavyHitterCCS16}
Z.~Qin, Y.~Yang, T.~Yu, I.~Khalil, X.~Xiao, and K.~Ren, ``Heavy hitter
  estimation over set-valued data with local differential privacy,'' in
  \emph{CCS}, 2016.

\bibitem{qin2017generating}
Z.~Qin, T.~Yu, Y.~Yang, I.~Khalil, X.~Xiao, and K.~Ren, ``Generating synthetic
  decentralized social graphs with local differential privacy,'' in \emph{CCS},
  2017.

\bibitem{smith2017interaction}
A.~Smith, A.~Thakurta, and J.~Upadhyay, ``Is interaction necessary for
  distributed private learning?'' in \emph{IEEE S \& P}, 2017.

\bibitem{bayes1763}
B.~Thomas, ``An essay towards solving a problem in the doctrine of chances,''
  \emph{Philosophical Transactions of the Royal Society of London}, 1763.

\bibitem{UberDP}
{Uber Differential Privacy}, ``\url{https://goo.gl/1Bcxsg}.''

\bibitem{Erlingsson:2014}
A.~K. {\'U}lfar~Erlingsson, Vasyl~Pihur, ``Rappor: Randomized aggregatable
  privacy-preserving ordinal response,'' in \emph{CCS}, 2014.

\bibitem{wang2018local}
S.~Wang, L.~Huang, Y.~Nie, P.~Wang, H.~Xu, and W.~Yang, ``Privset: Set-valued
  data analyses with local differential privacy,'' in \emph{INFOCOM}, 2018.

\bibitem{wang2017local}
S.~Wang, Y.~Nie, P.~Wang, H.~Xu, W.~Yang, and L.~Huang, ``Local private ordinal
  data distribution estimation,'' in \emph{INFOCOM}, 2017.

\bibitem{Ninghui:2017}
T.~Wang, J.~Blocki, N.~Li, and S.~Jha, ``Locally differentially private
  protocols for frequency estimation,'' in \emph{Usenix Security Symposium},
  2017.

\bibitem{wang2017locally}
T.~Wang, N.~Li, and S.~Jha, ``Locally differentially private heavy hitter
  identification,'' \emph{arXiv preprint arXiv:1708.06674}, 2017.

\bibitem{wang2018itemset}
T.~Wang, N.~Li, and S.~Jia, ``Locally differentially private frequent itemset
  mining,'' in \emph{IEEE S \& P}, 2018.

\bibitem{RR1965}
S.~L. Warner, ``Randomized response: A survey technique for eliminating evasive
  answer bias,'' in \emph{Journal of the American Statistical Association},
  1965.

\bibitem{ye2017optimal}
M.~Ye and A.~Barg, ``Optimal schemes for discrete distribution estimation under
  locally differential privacy,'' \emph{ISIT}, 2017.

\bibitem{zhang2018calm}
Z.~Zhang, T.~Wang, N.~Li, S.~He, and J.~Chen, ``Calm: Consistent adaptive local
  marginal for marginal release under local differential privacy,'' in
  \emph{CCS}.\hskip 1em plus 0.5em minus 0.4em\relax ACM, 2018.

\end{thebibliography}
}}

\appendices
\section{Proof of Theorem 1}
\label{prooftheorem1}
The proof of Theorem 1 leverages the following Lemma from probability theory~\cite{degroot2011probability}: 
\begin{lemma}
\label{lemma1}
Suppose $X$ is a random variable. $E(X-t)^2$ reaches its minimum value when $t=\mu$, where $\mu=E(X)$.
\end{lemma}

\begin{proof}

For any value of $t$, we have:
\begin{align}
E(X-t)^2 &=E(X^2-2tX+t^2) \nonumber \\
		&=E(X^2)-2t\mu+t^2
\end{align}
Then, we take derivative with respect to $t$:
\begin{align}
\label{partialn}
\frac{\partial E(X-t)^2}{\partial t}=-2\mu+2t
\end{align}
$E(X-t)^2$ reaches its minimum value when the derivative in Equation~\ref{partialn} is 0.  
By setting the derivative in Equation~\ref{partialn} to be $0$, we get $t=\mu$. 
Therefore, $E(X-t)^2$ reaches its minimum value when $t=\mu$.
\end{proof}

We can prove Theorem~\ref{optimality} via leveraging Lemma~\ref{lemma1}. Specifically, we can view the random variable $f$ in Theorem~\ref{optimality} as the random variable $X$ in Lemma~\ref{lemma1}, where the randomness is conditioned on that $\hat{f}=\hat{f}_i$. Therefore, $E((f' - f)^2 | \hat{f} = \hat{f}_i)$ reaches its minimum value when $f'=E(f|\hat{f} = \hat{f}_i)=\tilde{f}_i$.

\end{document}